\documentclass[10pt, a4paper]{article}
\usepackage{hyperref}
\usepackage{epstopdf}
\usepackage{stackrel}
\usepackage{graphicx}
\usepackage[utf8]{inputenc}
\usepackage{amssymb}
\usepackage{amsthm}
\usepackage{here}
\usepackage{amsmath, calc}
\usepackage[authoryear]{natbib}

\usepackage{xcolor}
\usepackage{subcaption}
\usepackage[boxruled,linesnumbered]{algorithm2e}
\usepackage{lineno,hyperref}
\usepackage{amsmath,amsfonts,amssymb}
\usepackage{mathrsfs}
\usepackage{amsthm}

\usepackage{enumerate}
\usepackage{graphicx}
\usepackage{pstricks,graphicx}
\usepackage{pstricks-add}
\usepackage{pst-eucl}
\usepackage{xcolor}
\usepackage[utf8]{inputenc}
\usepackage[T1]{fontenc}
\usepackage{verbatim}
\usepackage{algorithmic}
\newtheorem{lemma}{Lemma}[section]

\newtheorem{Proposition}{Proposition}[section]

\usepackage{tikz}
\usepackage{hyperref}
\usepackage{bm}
\usepackage[toc,page]{appendix}
\usepackage{subcaption}

\makeatletter
\newcommand*{\centerfloat}{%
  \parindent \z@
  \leftskip \z@ \@plus 1fil \@minus \textwidth
  \rightskip\leftskip
  \parfillskip \z@skip}
\makeatother

\begin{document}


\title{A random model for multidimensional fitting method}

\author{Hiba Alawieh\\ 
Laboratoire Paul Painlev\'e, UFR de math\'ematiques, \\
Universit\'e Lille 1, $59650$ Villeneuve d'Ascq, France.\\
~\\
Fr\'ed\'eric Bertrand\\
IRMA, Universit\'e de  Strasbourg \& CNRS, \\
67084, Strasbourg, France.\\
~\\
Myriam Maumy-Bertrand\\
IRMA, Universit\'e de  Strasbourg \& CNRS, \\
67084, Strasbourg, France.\\
~\\
Nicolas Wicker\\
Laboratoire Paul Painlev\'e, UFR de math\'ematiques, \\
Universit\'e Lille 1, $59650$ Villeneuve d'Ascq, France.\\
~\\
Baydaa Al Ayoubi\\
Département de math\'ematiques appliqu\'ees,\\
Faculté des sciences, Université libanaise, \\
Beyrouth, Liban.
}

\maketitle

\begin{abstract}
Multidimensional fitting (MDF) method is a multivariate data analysis method recently developed and based on the fitting of distances. Two matrices are available: one contains the coordinates of the points and the second contains the distances between the same points. The idea of MDF is to modify the coordinates through modification vectors in order to fit the new distances calculated on the modified coordinates to the given distances. In the previous works, the modification vectors are taken as deterministic variables, so here we want to take into account the random effects that can be produce during the modification. An application in the sensometric domain is also given.

\end{abstract}




\section{Introduction}

Multidimensional data matrices are encountered in many disciplines of science as in biological domain studying the gene expression data \cite{gene1,gene2},  the geographical domain analysing the spatial earthquake data \cite{geo}, in financial market data for portfolio construction and assessment \cite{fin} and many others domains. The complexity of  matrices can be observed for example when a variable is informative only for a subset of  data which   render the application of some statistical methods too hard. Therefore, multidimensional data analysis \cite{Mardia1979} refers to the process of summarizing data across multiple dimensions and presenting the results in a reduced dimension space. Several well-known methods exist which perform dimension reduction as principal component analysis, multiple correspondence analysis, multidimensional scaling, and others.  A new method of multidimensional data analysis called multidimensional fitting (MDF) was introduced and studied in \cite{MDF} and \cite{PMDF}. This method is a new method of fitting distances used in data analysis and requires two observed matrices, the target matrix and the reference matrix. The idea of MDF method is to modify the coordinates given by the target matrix in order to fit the new distances calculated on these modified coordinates to distances given by the reference matrix. In \cite{MDF} and \cite{PMDF}, the authors consider the displacement vectors as deterministic vectors and the random effect that can be produced during the modification and can affect the interpretation of the modification significance is not taken into account.

In this article we want to introduce the random effect in the model of MDF method and  to find then the real displacement vectors. Here, the minimization of the mean square error between the new distances and the reference distances performed in the deterministic model of MDF, to obtain the optimal values of displacement vectors, cannot be applied for the random model as the objective function here is a random variable. Therefore, we want to use different ways to find these vectors.

\noindent First of all, the random model of MDF is presented in Section $\ref{sec1}$, then in Sections $\ref{SecMin}$ and $\ref{secSim}$, two ways to obtain the optimal values of displacement vectors are illustrated. After that, an application in the sensometric domain has been presented in Section $\ref{secApp}$ in order to fit the sensory profiles of products to consumers preferences of these products. Finally, we conclude our work in Section $\ref{SecConc}$.

\section{The random model of multidimensional fitting} \label{sec1}
Let us consider a target matrix $X=(X_1|\dots|X_n)$ which contains the coordinates of $n$ points  in $\mathbb{R}^{^p}$ with $p>1$ in $\mathbb{R}$ and a reference matrix $D=\{d_{ij}\}$ which contains the distances between the same $n$ points.
We note $f$ the modification function such as: \[f(X_i)=X_i+L_i,\] where for $i=1,\dots,n$, the vectors $X_i$ and $ L_i$ in ${\mathbb{R}}^{^{p}} $  are, respectively, the coordinate and the displacement vectors of point $i$.

\noindent The problem of MDF is a mean square error minimization and the error, noted $\Delta$, defined by:
\[\Delta=\sum_{1\leq i<j\leq n} \left( d_{ij}-a \delta_{ij}\right)^2
\]
where $\delta_{ij}=\lVert f(X_i)-f(X_j) \rVert_{_2}=||X_{i}+L_{i}-X_{j}-L_{j}||_{_2}$ and $a$ a real scaling variable. We note: $e_{ij}=(d_{ij}-a \delta_{ij})^2\cdot$
    
\noindent Owing to the negligence of random effects that can occur during modification, the interpretation of the displacements can be erroneous. Thereby, to tackle this problem we introduce the random effects in the modification function. So, the new modification function is given by: \[f(X_i)=X_i+\theta_i+\varepsilon_i,\]
where  $\theta_i$ and $\varepsilon_i$ in $ {\mathbb{R}}^{^{p}}$ are, respectively, the fixed and random part of modification.

\noindent 	Contrary to what has been seen above,  $\delta_{ij}$ here is a random variable and not a deterministic value, for all $1 \leq i<j\leq n$, so the error $\Delta$ cannot be minimized directly.
Deterministic and stochastic optimization are presented to find the optimum value of vectors $(\theta_1,\dots,\theta_n)$:
\begin{itemize}
\item[1-] Deterministic optimization: by minimizing a function depending on vectors $\theta_1,\dots,\theta_n$ with consideration that the components of vectors $\varepsilon_i$, for all $i=1,\dots,n$, are independently and identically normally distributed.
\item[2-] Stochastic optimization: by simulating the error $\Delta$ where the components of vectors $\varepsilon_i$ for all $i=1,\dots,n$ are dependent and not normally distributed.
\end{itemize}

\section{Calculation of $(\theta^*_1,\dots,\theta^*_n)$ by minimization} \label{SecMin}

\noindent In this section, we suppose that the components of vector $\varepsilon_i$ denoted $\varepsilon_{ik}$, for all $i=1,\dots,n$ and $k=1,\dots,p$, are $p$-dimensional independent and identically distributed random variables where the vector $\varepsilon_i$ is multivariate normally distributed with mean $\mathbb{E}(\varepsilon_i)=\textbf{0}$ (the vector $\textbf{0}$ in $\mathbb{R}^{^{p}}$ is the null vector) and variance $\mathbb{V}ar(\varepsilon_i)=\sigma^2 I_p$ ($\sigma$ is a strictly positive value to be fixed and $I_p$ is the identity matrix).

\noindent We note $\Gamma$ a $n \times n$ matrix that contains the distances between the points after modification. The objective function of the minimization problem called $g({D}, \Gamma)$ can be expressed in different forms. We cite below some of them:
$$\begin{array}{cccc}
g: & \mathcal{M}_{n\times n}(\mathbb{R}) \times \mathcal{M}_{n\times n}(\mathbb{R}) & \longmapsto & \mathbb{R}\\
 & ({D}, \Gamma) & \longmapsto &  \begin{cases}
\begin{array}{c}
\mathbb{E}(\Vert {D}-a \Gamma \Vert_2^2) \\
\text{ }\text{med}(\Vert {D}- a \Gamma \Vert_2^2)\\
\text{ } \text{minmax}(\Vert {D}-a \Gamma \Vert_2^2)\\
\text{  }\Vert {D}-a \ \mathbb{E}(\Gamma) \Vert_2^2 \\
\text{  }\Vert {D}-a \ \text{med}(\Gamma) \Vert_2^2 \\
\text{  }\Vert \text{minmax}({D}-a  \Gamma) \Vert_2^2. \\
\end{array}\end{cases}
\end{array}$$

\noindent In our work, we are interested to take $g({D}, \Gamma)=\mathbb{E}(\Vert {D}-a\Gamma \Vert_{_2}^2)$. The expression of  $\Vert {D}-a\Gamma \Vert_{_2}^2$ noted $\Delta$ (as the mean square error cited above) can be rewritten as: \[\Delta=\displaystyle \sum_{1\leq i < j \leq n} (d_{ij}-a \Vert X_i+\theta_i+\varepsilon_i -X_j -\theta_j -\varepsilon_j \Vert_{_2})^2\cdot \]

\noindent The problem here is to find the vectors ($\theta^*_1, \dots, \theta^*_n$) such that the minimum of $\mathbb{E}(\Delta)$ under  ($\theta_1, \dots, \theta_n$) is reached. The initial optimization problem (P$_0$) is defined by:
\[ (\text{P}_0): \underset{\theta_1,\dots,\theta_n \in \mathbb{R}^{^p}}{\text{min }} \mathbb{E}(\Delta)\cdot \]

\noindent  The optimal solution obtained from (P$_0$) is a solution assigns the smallest value to $\mathbb{E}(\Delta)$  but moves too many points. So, it is a good solution from  minimization standpoint, but awkward from parsimony standpoint.

\noindent A new optimization problem is presented to find the optimal vectors $(\theta^*_1,\dots,\theta^*_n)$ by taking into account the minimization of the expectation of $\Delta$ and the parsimonious choice of displacements. A natural approach to obtain such sparsity solution is to use the number of non-zero displacements as a penalty. So, a penalty term can be defined, using $\ell_0$-norm,  as $\sum_{i=1}^n\rVert \theta_i \lVert_{_0}$  where \[\rVert . \lVert_{_0}=\# (i=1,\dots,n; k=1,\dots,p|\ \theta_{ik}=0)\] is the $\ell_0$ norm which measures the parsimony of the displacements of points. Thus, a new optimization problem called (P$_1$) is given by:
\[(\text{P}_1): \underset{\theta_1,\dots,\theta_n \in \mathbb{R}^{^p}}{\min} \hspace{2mm}\mathbb{E}(\Delta)+\eta \sum_{i=1}^n\lVert\theta_{i}\rVert_{_{0}},\\\\
\] 
with $\eta$ is a positive regularization parameter to be chosen. It controls the trade-off between the minimization of the expectation of the error and the use of a parsimonious number of displacements.

\subsection{Choice of  regularization parameter }
\noindent In different penalization problems as the penalized regression or penalized likelihood methods for high dimensional data analysis \cite{PL}, the choice of regularization parameter is always crucial to lead good results attached to the problem at hand. Different methods have been introduced to find the good value of this parameter (see \cite{pen1},\cite{pen2}). Some practical approaches consist in comparing  different models using a sequence of penalization parameter and then choose the best of them using some model selection criteria like Cross-validation (CV) (\cite{cv1}, \cite{cv2}), Akaike Information Criterion (AIC) \cite{aic} and Bayes Information Criterion (BIC) \cite{bic}. 

\noindent In our model, the choice of the value of $\eta$ is related to the number of displacements. With the same practical concept as the approaches presented in the literature, we want to  solve the optimization problem (P$_1$) by taking different values of $\eta$, and as our problem is related to the  number of displacements so  we choose a value of $\eta$ that takes into account the number of points that must be modified in our data to fit the references distances. This number of points can be computed from the data or fixed by the number of displacements that we want to perform. So, the chosen number of displacements can be taken by two ways:
\begin{itemize}
\item[1-] through the posed problem,
\item[2-]  using the data.
\end{itemize}   

\noindent Obviously, first way is trivial. Indeed, it is sufficient an user or a company choose a fixed number of displacements that wish perform  to find the desirable solution. Accordingly, fixing the number of displacements can be interesting to companies because in some cases a displacement can be difficult  and expensive therefore it is suitable for them to fix at the beginning  the number of displacements.  
For the second way, the number of displacements can be calculated using the data. Therefore, a criterion of points selection defined below is used to choose the number of displacements.

\subsubsection{Criterion for selection points}\label{Secrho}
The number of displacements is related to the number of points that are misplaced in their initial configuration and need movements  to fit their reference distances. For that, we have developed a simple criterion based on the error calculated on the initial data before data modification. This criterion for selection of the points is denoted $\rho_i$.

\noindent Indeed, for $i=1,\dots,n$ and $j=i+1,\dots, n$, we calculate the following difference:
 \[r_{ij}=(d_{ij}-a\lVert X_i-X_j\rVert_{_2})^2.\] 
\noindent Note that $r_{ij}$ is equivalent to $e_{ij}$ with $L_i=L_j=0$.  

\noindent Then, for each $i=1,\dots, n$, the criterion for selection points is defined as: \[\rho_{i}=\frac{\displaystyle{\sum_{m=1,m\neq i}^{n}r_{im}}}{\displaystyle{\sum_{1 \leq i<j\leq n}^{n} r_{ij}}}\cdot\]

\noindent The values of $\rho_i$ are between $0$ and $1$ so, for fixed value $\varrho \in [0,1]$  which is chosen through the value of $\rho_i$:
\begin{itemize}
\item if $\rho_i\leq \varrho$, then $i$ is considered as correctly placed point,
\item else,  $i$ is considered as misplaced point.
\end{itemize}
\vspace{5mm}

Now, in order to verify the interest of the modification of coordinates so as to approximate the distances, we want to perform a statistical test on the displacement vectors ($\theta^*_1, \dots,\theta^*_n$). 

\subsection{Statistical test for the displacement vectors $(\theta^*_1,\dots,\theta^*_n)$} \label{Sectest}
In this section, we want to present a statistical test for the displacement vectors  $\theta_i$ for all $i=1,\dots,n$. This test is based on the hypothesis of displacements significance. Recall the error:
\begin{equation}
\Delta=\displaystyle \sum_{1\leq i<j\leq n} \left(d_{ij}-a\lVert X_{i}+\theta_{i}+\varepsilon_{i}-X_{j}-\theta_{j}-\varepsilon_j \rVert_{_{2}} \right)^2. 
\end{equation}\label{eqE1} 

\noindent We note $\Delta_0$ the initial error given by: \[ \Delta_0=\sum_{1 \leq i <j \leq n} (d_{ij}-a \Vert X_i+\varepsilon_i-X_j-\varepsilon_j\Vert_{_2})^2.\]

\noindent The two hypothesis of the statistical test are:

\[\begin{cases}
\begin{array}{cc}
&(\mathcal{H}_{0}):\\
\\
&(\mathcal{H}_{1}):
\end{array} & \begin{array}{cc}
\text{For } (\theta^*_{1}, \dots, \theta^*_n) \text{ such that  }  d_{ij}=a\lVert X_i+\theta_i-X_j-\theta_j\rVert_{_2},\text{ for all }(i,j),\\
\text{ the initial error } \Delta_0 \text{ and the error }\Delta \text{ calculated from }\\
\text{the vectors  }(\theta^*_1,\dots,\theta^*_n) \text{ have the same distribution.}\\

\text{ The initial error }\Delta_0 \text{ and the error } \Delta \text{ calculated from the vectors }\\
(\theta^*_1,\dots,\theta^*_n) 
\text{have not the same distribution.}
\end{array}\end{cases}\] 

\noindent The error $\Delta$ is the test statistic and the decision rule is the following:
\[ \text{Rejection of } (\mathcal{H}_0) \Leftrightarrow \mathbb{P}_{\mathcal{H}_0}[\text{Reject }(\mathcal{H}_0)]\leq \alpha\Leftrightarrow \mathbb{P}_{\mathcal{H}_0}[\Delta\geq \Delta_c]\leq\alpha\cdot\]  

\noindent To perform this test, we use the Bienaymé-Tchebychev inequality:
\[\forall \gamma>0, \hspace{2mm }\mathbb{P}[|\Delta-\mathbb{E}(\Delta)| \geq \gamma]\leq \frac{\mathbb{V}ar(\Delta)}{\gamma^2}\cdot\] 

\noindent Moreover, we suppose that the random effect is injected in the observation so instead of observing $X_i$ we observe $X_i+\varepsilon_i$. Then, by choosing  $\gamma=|\Delta_0-\mathbb{E}(\Delta)|$, the ratio $\displaystyle \frac{\mathbb{V}ar(\Delta)}{\gamma^2}$ can be considered as $p$-value. So, if it is small than $\alpha$ then we reject the null hypothesis $(\mathcal{H}_0)$ with $\alpha$ is the error of type I. 

\noindent  computation of expectation and variance of error $\Delta$ are done in appendix. 
The results obtained in the  appendix are valid for any values of $\theta_1,\dots,\theta_n$. Therefore, under the hypothesis $(\mathcal{H}_0)$ it is sufficient to replace $\Vert X_i+\theta_i-X_j-\theta_j\Vert_{_2}$ by $d_{ij}$.

\subsection{The optimization problem}
Once hypothesis  $(\mathcal{H}_0)$ is rejected, the vectors ($\theta^*_1,\dots,\theta^*_n$) can be calculated by solving problem (P$_1$).

\noindent Using the results of appendix $\ref{App2}$, the expectation of the error $\Delta$ has been calculated from the expectation of the non-central chi-squared and  non-central chi distribution. 
\begin{Proposition}
The expectation of the error $\Delta$ is:
 \[ \displaystyle \mathbb{E}(\Delta)=\sum_{1\leq i<j\leq n} \left[ d_{ij}^2 +2a^2\sigma^2(p+\lambda^2_{ij})-2\sqrt{\pi}a \sigma d_{ij} \mathcal{L}_{\frac{1}{2}}^{\frac{p}{2}-1}\left(-\frac{\lambda_{ij}^2}{2}\right) \right], \label{eqEsp}\]
where $\displaystyle \lambda_{ij}=\frac{1}{\sqrt{2}\sigma}\lVert X_i+\theta_i-X_j-\theta_j\rVert_{_2}$ and $\mathcal{L}_{\nu}^{\gamma}(x)$ is the generalized Laguerre polynomial \cite{laguerre}.
\end{Proposition}
 
\noindent The optimization problem (P$_1$) can rewritten as:
\[(\text{P}_1): \underset{\theta_1,\dots,\theta_n \in \mathbb{R}^{^p}}{\min} \hspace{2mm}a^2\lVert X_i+\theta_i-X_j-\theta_j\rVert^2_{_2}-2\sqrt{\pi}a \sigma d_{ij} \mathcal{L}_{\frac{1}{2}}^{\frac{p}{2}-1}\left(-\frac{\lVert X_i+\theta_i-X_j-\theta_j\rVert^2_{_2}}{4\sigma^2}\right) +\eta \sum_{i=1}^n\lVert\theta_{i}\rVert_{_{0}},\\\\
\]

\section{Calculation of $(\theta^*_1,\dots,\theta^*_n)$ by simulation} \label{secSim}
In this section, we suppose that the $p$ components of $\varepsilon_i$ are dependent or/and not necessarily normally distributed so the application of chi-squared and chi distributions becomes impossible. Therefore, we want to present an algorithm noted Algorithm $1$ which allows us to find the optimal vectors $\theta^*_1,\dots,\theta^*_n$ using Metropolis-Hastings \cite{MH1}.

\subsection{Simulation tools}\label{seKT}
Different tools used in algorithm $1$ and associated to the generation of vectors $\theta_1,\dots,\theta_n$ in order to minimize the error $\Delta$  are presented in the follow.
 
\subsubsection{Identification of misplaced and correctly placed sets}\label{SelP}

The set of points can be divided into two subsets:
\begin{itemize}
\item The first, noted $W$, having size equal to $n_W$. This subset contains the points that are correctly placed and should not be moved.
\item The second, noted $M$, having size equal to $n_M$. This subset contains the points that are misplaced and must be moved. 
\end{itemize}

\noindent The criterion for points selection  $\rho_i$ presented in Section $\ref{Secrho}$ is used to construct these two subsets.

\subsubsection{Movement of set $M$}
The subset $M$ contains the misplaced points that must be moved  in order to fit the reference distances. In this section, the work is concentrated to find movements for the subset $M$ approaching as possible as the distances calculated after movements to the reference distances. 
The movements that can be applied to $M$ are translation, scaling and rotation. The scaling movement is not interest in our study as the subsets $W$ and $M$ are in the same scale seen that are derived from the same set of points and the scaling variable $a$ presented in the optimization problem of MDF is kept. Moreover, we suppose that the points inside  $M$ are well concentrated so that the rotation movement can be neglected. That is why we are just interested on the translation movement. The translation of $M$ through a vector $B\in \mathbb{R}^{^{p}}$ can be shown as the translation of each points in $M$. So, the translation of a point $j\in M$ is given by: $Y_j+B$ where $Y_j \in \mathbb{R}^{^p}$ is the coordinate vector of point $j$.

\noindent The translation movement is performed in such a way to approach the distances calculated after translation to the distances given by the reference matrix. Thus, find the vector $B$ return to solve the following optimization problem:
\[(\mathcal{P}): \underset{B \in \mathbb{R}^{^{p}} }{\min} \sum_{i\in W}\sum_{j\in M}  \left(d_{ij}^{2}-a\parallel X_{i}-Y_{j}-B \parallel_{_2}^{2}\right)^{2}.\]

\noindent In order to simplify the problem $(\mathcal{P})$, we suppose that for all $i\in W$ and $j\in M$, $d_{ij}^{2}-a \parallel X_{i}-Y_{j}-b\parallel_{_2}^{2} \geq 0$ and the problem $(\mathcal{P})$ becomes:
\[(\mathcal{P}_{1})\begin{cases}
\begin{array}{c}
\\
\text{s.t}
\end{array} & \begin{array}{c}
\underset{B  \in \mathbb{R}^{^{p}}}{\min}\displaystyle{\sum_{i\in W}\sum_{j\in M}}\left(d_{ij}^{2}-a \parallel X_{i}-Y_{j}-B\parallel^{2}_{_2}\right)\\
\forall i\in W \text{ and }  j\in M, \ d_{ij}^{2}-a \parallel X_{i}-Y_{j}-B\parallel_{_2}^{2} \geq 0\cdot
\end{array}\end{cases}\]
\textbf{Relaxation of problem $(\mathcal{P}_{1})$:}
The following problem $(\mathcal{P}_{2})$ can easily solved and provide a starting point to resolve $(\mathcal{P}_{1})$. We have:

\begin{eqnarray}
(\mathcal{P}_{2}): & \displaystyle{ \sum_{i\in W}\sum_{j\in M}\left(d_{ij}^{2}-a \parallel X_{i}-Y_{j}-B\parallel_{_2}^{2}\right)}=0 \label{eqnul} \\
\Leftrightarrow & \displaystyle{\sum_{i\in W}\sum_{j\in M}}a \left(\parallel X_{i}-Y_{j}\parallel_{_2}^{2}+\parallel B\parallel_{_2}^{2}-2\langle X_{i}-Y_{j},B\rangle\right)-\displaystyle{\sum_{i\in W}\sum_{j\in M}}d_{ij}^{2}=0\nonumber\\
\Leftrightarrow & a\parallel B\parallel_{_2}^{2}-2a\frac{\displaystyle{\sum_{i\in W}\sum_{j\in M}}\langle X_{i}-Y_{j},B\rangle}{n_{W}n_{M}}=\frac{\displaystyle{\sum_{i\in W}\sum_{j\in M}}d_{ij}^{2}}{n_{W}n_{M}}-a \frac{ \displaystyle{\sum_{i\in W}\sum_{j\in M}} \Vert X_{i}-Y_{j}\Vert_{_2}^{2}}{n_{W}n_{M}} \cdot\nonumber
\end{eqnarray}
Hence,
\begin{eqnarray}
a\left\| B-\frac{\displaystyle{\sum_{i\in W}\sum_{j\in M}} (X_{i}-Y_{j})}{n_{W}n_{M}}\right\|_2^{2}=&\frac{\displaystyle{\sum_{i\in W}\sum_{j\in M}}d_{ij}^{2}}{n_{W}n_{M}}-a\frac{\displaystyle{\sum_{i\in W}\sum_{j\in M}}\parallel X_{i}-Y_{j}\parallel_{_2}^{2}}{n_{W}n_{M}}\nonumber\\
&+a\frac{\left( \displaystyle{ \parallel \sum_{i\in W}\sum_{j\in M}} (X_{i}-Y_{j})\parallel_{_2}\right) ^2} {n^2_{W}n^2_{M}}\nonumber\cdot
\end{eqnarray}

\noindent So,
\begin{equation}
\left\| B-\frac{\displaystyle{\sum_{i\in W}\sum_{j\in M}} (X_{i}-Y_{j})}{n_{W}n_{M}}\right\|_2^{2}\leq r^2, \label{eq2}
\end{equation}
\noindent with \[r^2=\left| \frac{\displaystyle{\sum_{i\in W}\sum_{j\in M}}d_{ij}^{2}}{a \ n_{W}n_{M}}-\frac{\displaystyle{\sum_{i\in W}\sum_{j\in M}}\parallel X_{i}-Y_{j} \parallel_2^{2}}{n_{W}n_{M}}+\frac{\left( \displaystyle{ \parallel \sum_{i\in W}\sum_{j\in M}} (X_{i}-Y_{j})\parallel_2\right) ^2} {n^2_{W}n^2_{M}} \right|\cdot\]

\noindent Using inequality (\ref{eq2}), we can conclude that the optimal solution of the vector $B$ belongs to an hypersphere $(S)$ centered in $ C=\displaystyle \frac{\sum_{i\in W} \sum_{j\in M} (X_i-Y_j)}{n_W n_M} $ with radius $r$. As Equation (\ref{eqnul}) is never equal to zero
during optimization, so we take it smaller than certain real value. 
Therefore, the optimal solution is guened to belong to an hypersphere 
 $(S^\xi)$ with same center $C$ as hypersphere $(S)$ but with radius equal to $r\pm\xi$ with $\xi$ a small value in $\mathbb{R}^{+}$. 

\noindent We suppose that vector $B$ is uniformly distributed between $0$ and $B_{max}$, so it is necessary to find the maximal value $B_{max}$. This value is geometrically determined using  Figure $\ref{figB}$.  Starting with a null value of vector $B$, $B$ moves uniformly on the line $(d)$ passing by $O$ (the point where the vector $B$ is null) and $C$  to reach its maximum at the point $A$, the far intersection between $(d)$ and hypersphere $(S^\xi)$, hence the uniqueness of $A$.

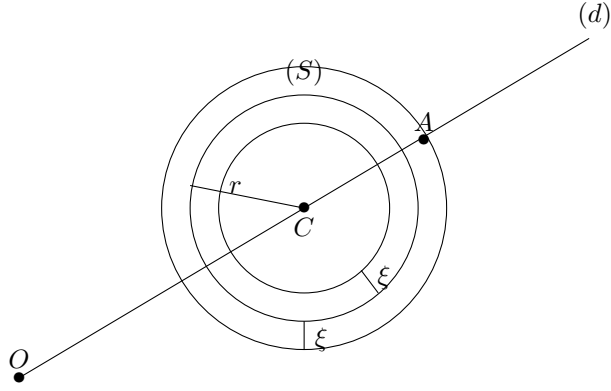
\begin{figure}[htp]
\begin{center}
\begin{tikzpicture}[scale=3/4]
\draw (0,0) circle (2.5cm);
\draw (0,0) circle (2cm);
\draw (0,0) circle (1.5cm);
\draw (0,0) node[below ]{$C$} ;
\draw (-5,-3) -- (5,3);
\draw (0,0) -- (-2,0.4);
\draw (1,-1.1) -- (1.3,-1.5);
\draw (0,-2) -- (0,-2.5);
\draw (0.3,-2.7) node [above] {$\xi$};
\draw (1.4,-1.6) node [above] {$\xi$};
\draw (-1.2,0.1) node [above] {$r$};
\draw (-5,-3) node [above] {$O$};
\draw (2.1,1.2) node [above] {$A$};
\draw (5.1,3) node [above] {$(d)$};
\draw (0,2) node [above] {$(S)$};
\foreach \Point in {(0,0), (2.1,1.2), (-5,-3)}{
    \node at \Point {\textbullet};
}

\end{tikzpicture}
\end{center}
\caption{ \textit{Illustration of the determination of vector $B$ in $\mathbb{R}^2$.  The maximal solution of $B$ is located at $A$. The values of vector $B$ moves uniformly on the segment $[OA]$.}} \label{figB}
\end{figure}

\noindent To calculate the maximal solution of vector $B$, it is needed to find the far intersection of the line $(d)$ with hypersphere $(S^{\xi})$. The line $(d)$ has as direction vector the vector $OC$. So, the parametric equation of $(d)$ is equal to:
\begin{equation}\label{eqSS}
B=(t+1)C\cdot
\end{equation}

Furthermore, we have:
\[\begin{cases}
\begin{array}{c}
B=(t+1)C\\
\parallel B-C\parallel_{_2}^{^2}=(r+\xi)^{2}.
\end{array}\end{cases}\]

\noindent The intersection between $(d)$ and $(S^{\xi})$ gives:
\begin{eqnarray}
\parallel(t+1)C-C\parallel_{_2}^{2}=(r+\xi)^{2}&\nonumber\\
\parallel tC \parallel_{_2}^{2}=(r+\xi)^{2}&\nonumber\\
t^2 \parallel C \parallel_{_2}^{2}=(r+\xi)^{2} &\nonumber\\
t=\pm \frac{r+\xi}{\parallel C \parallel_{_2}}\cdot& \nonumber
\end{eqnarray}

\noindent We are interested by the farthest intersection, thus we take $\displaystyle t= \frac{r+\xi}{\parallel C \parallel_{_2}}\cdot$ By replacing $t$ in Equation (\ref{eqSS}), we obtain:
\[B_{max}=\left( \frac{r+\xi}{\parallel C \parallel_{_2}}+1 \right) C.\]

\noindent The values of $B$ can be proposed uniformly on the segment $[OA]$, so \[B \leadsto \mathcal{U}\left(0,B_{max}\right).\]

\subsubsection{Movement vectors generation } \label{MovSec}
\noindent In practice, for all $k=1, \dots, n$, we suppose that  $M$ contains one point noted $l$. The choice of this point is made by a multinomial distribution $\mathcal{M}(1,\rho_1, \dots, \rho_n)$ where $\rho_k$ for $k=1,\dots,n$ is as defined in section $\ref{Secrho}$. 

\noindent At instant $t$, a point $l$ chosen as misplaced point must occur a movement through the uniform distribution such as $\mathcal{U}\left(0;(\frac{r_l+\xi}{\parallel C_l \parallel}+1)C_l\right)$ with $C_l$ and $r_l$ are, respectively, the center and the radius of hypersphere $(S_l^{\xi})$ obtained by taking $M=\{l\}$. Thus, the movement of the point $l$ is equal to the movement at instant $t-1$ plus the new movement obtained by uniform distribution. Hence, we can write:  \[\theta^*_l=\theta^{t-1}_l+B_l,\]
with $B_l$ is a generation value of the uniform distribution $\mathcal{U}\left(0,B^l_{max}\right)$. 
Noted that the equation of hypersphere given by (\ref{eq2}) in each instant depends of misplaced point $l$.  

\noindent We note
${\Theta}$ the sequence of $n$ generated vectors in ${\mathbb{R}}^{^{p}}$ defined by: 
${\Theta}=\left(
{\theta}_{1},\dots,{\theta}_{n}
\right)$. Therefore, the passage from $\theta_l^{t-1}$ to $\theta_l^{t}$ occurs in a way to move: \[\Theta^{t-1}=(\theta_1^{t-1}, \dots, \theta_l^{t-1}, \dots, \theta_n^{t-1})\]  to \[\Theta^t=(\theta_1^{t-1}, \dots, \theta_l^{t}, \dots, \theta_n^{t-1}).\]

\subsubsection{Proposal distribution}\label{seKTR}
A proposal distribution is needed in the Metropolis-Hastings algorithm defined below. This distribution is constructed by calculating the probability to pass from $\Theta^{t-1}$ to a new generate value of $\Theta$ denoted $\Theta^*$ and it is equal to the probability to choose a point $l$ multiplied by the probability of the movement of this point. 
So, the proposal distribution noted $q$ is given by:
\[q(\Theta^{t-1}\longrightarrow \Theta^{*})=\rho_{l} \times \frac{1}{\frac{r_{l}+\xi}{\parallel C_{l}\parallel}+1}\] with $l$ is the chosen point. 

\noindent We can easily see that this proposal distribution is a probability density function as $\sum_{i=1}^n \rho_i=1$.

\subsection{Calculation  of $(\theta_1,\dots,\theta_n)$ using Metropolis-Hastings algorithm}
\noindent We consider that the component of vector $ \varepsilon_i$ are dependent such that $ \varepsilon_i \rightsquigarrow \mathcal{N}_p(\textbf{0},\Sigma)$, with $\Sigma$ is the covariance matrix. 

\noindent The Metropolis-Hastings algorithm allows us to build a Markov chain with a desired stationary distribution \cite{MH1},\cite{MH2}. The proposal distribution here is related to the choice of vectors $\theta_i$ for $i=1,\dots,n$ and it is given in paragraph \ref{seKTR}. The target distribution is given by:
\[\pi(\Theta,\varepsilon)\propto{\exp\left(\frac{-E(\Theta)}{T}\right)\cdot h(\varepsilon)}\] where $E$ is an application  given by:
\[\begin{array}{cccc}
E: & \mathcal{M}_{n \times p} & \longmapsto & \mathbb{R}\\
 & \Theta=(\theta_{1},\ldots,\theta_{n}) & \longmapsto & \displaystyle E(\Theta)=\sum_{1 \leq i<j \leq n} \left(d_{ij}-a \Vert X_i+\theta_i-X_j -\theta_j\Vert_{_2} \right)^2,
\end{array}\]
\noindent and $h$ is the density function of the normal distribution $ \mathcal{N}_p(\textbf{0},\Sigma)$.  
The variable $T$ is the temperature parameter, to be fixed according to the value range of $E$.

\noindent The algorithm of Metropolis-Hastings is given as follows:\\
\begin{algorithm}[H] \label{algo2}
\caption{}
\begin{algorithmic}
\STATE{Initialization: $\Theta_0=(\theta_1|\dots|\theta_n)=(\textbf{0}| \dots| \textbf{0})$.}
\STATE{Calculate the ratios $\rho_1,\dots,\rho_n$. }
\FOR{$t=1$ to $N_1$} 
\STATE {Generate a point $l$ using multinomial distribution $\mathcal{M}(1,\rho_1,\dots,\rho_n)$.}
\STATE{Generate a vector $B_l$ using the uniform distribution $\mathcal{U}(0,B^l_{max})$ with $B^l_{max}=\left( \frac{r_l+\xi}{\parallel C_l \parallel_{_2}}+1 \right) C_l$}.
\STATE{Generate the vector $\theta_l^*=\theta_l^{t-1}+b_l$ and for all $i\in[1,\dots,n]-\{l\}$ take\\ the vectors $\theta_i^*$ equal to $\theta_i^{t-1}$. }
\STATE{Generate the vectors $\varepsilon^*_i$ using normal distribution $\mathcal{N}(\textbf{0},\Sigma)$ for $i=1,\dots,n$.}
\STATE {Calculate  $\displaystyle \alpha=\min\{1,\frac{\exp\left(\frac{-f(\Theta^*)}{T}\right) h(\varepsilon^*) q(\Theta^*\rightarrow \Theta^{t-1})}{\exp\left(\frac{-f(\Theta^{t-1})}{T}\right) h(\varepsilon^{t-1}) q(\Theta{^{t-1}}\rightarrow \Theta^*)}\}$}.
\STATE{Generate $u\thicksim \mathcal{U}(0,1)$}.
\IF {$u \leq \alpha$}
\STATE{ $\Theta^t=\Theta^*$}
\ELSE
\STATE{$\Theta^t=\Theta^{t-1}$}.
\ENDIF
\ENDFOR
\STATE{Choose $\Theta$ that gives the minimum value of error $\Delta$}
\end{algorithmic}
\end{algorithm}

\noindent \textbf{Remark:} The error $ \varepsilon_i$ can be distributed through any other distribution other than the Gaussian distribution, so it is sufficient to generate the vector $\varepsilon_i$ using this distribution instead the normal distribution in algorithm 1.

\section{Application}\label{secApp}

This random model of multidimensional fitting has been applied in the sensometrics domain. This relatively young domain concerns the analysis of data from sensory science in order to develop a product by linking sensory attributes to ingredients, benefits, values and emotional elements of the brand to design products that meet the sensory quality preferences of sensory-based consumer segments \cite{Tortila}.
This analysis of product characteristics among consumers gives an overview of the positive and negative aspects of products and aid the companies to better meet consumer tastes.
So, the problem here is to fit the consumers scores to the product configuration given by the experts in order to find the ideal sensory profile of a product. Thus, two matrices are at disposal, one contains the consumer scores of products and the second the sensory profile of products given by the experts.  

\noindent Several modelling techniques have been applied in sensory analysis domain like preference mapping which is the must popular of them. They can be divided into two methods: internal and external analysis \cite{Sen}. These methods have as objective to visually assess the relationship between the product space and patterns of preference \cite{Sen2}. In our application, we want to use the random model of multidimensional fitting to match as well as possible the sensory profile to the consumers preference of a product. White corn tortilla chips and muscadine grape juice data sets are used in our application. 

\subsection{Data description} 
White corn tortilla chips data set has been studied in \cite{Tortila} where $80$ consumers rated $11$ commercially available toasted white corn tortilla chip products for overall liking, appearance liking, and flavor liking.  The names of these $11$ tortilla chip products and their labels are given in the Table $\ref{tabname1}$.
Moreover, a group of $9$ Spectrum trained panelists evaluated apperance, flavor and texture attributes of tortilla chips using the Spectrum Method \cite{Sen} (Sensory Spectrum Inc., Chantham, NJ, U.S.A.). This data set is available at "\url{http://www.sensometric.org/datasets}" and it is composed from consumers notes table and panelists notes table. The first table is constructed after asked each consumer to evaluate liking, appearance, flavor and texture of each  tortilla chips sample on a 9-point dedonic scale and the saltiness on 5-point 'JustAboutRight' (JAR) (for more information about the scale, visit "\url{ http://www.sensorysociety.org/knowledge/sspwiki/Pages/The209-point20Hedonic20Scale.aspx}"). The second table is obtained after the evaluation of the $9$ panelists for flavor, texture and appearance attributes of all the chips and after that the calculation of the average score for each attribute. The total number of attributes studied in panelists notes table is $37$, we note some of them: sweet, salt, sour, lime, astringent, grain complex, toasted corn, raw corn,	masa, toasted grain\dots

\noindent The application of our method requires a target and reference matrices. The target matrix is given by the panelists notes table, so the dimension of this matrix is $11 \times 37$ and the reference matrix is a matrix of dimension $11 \times 11$ and contains the Euclidean distances between the different tortilla chip samples calculated using the consumers notes table.

\noindent Muscadine grape juice data set is well studied in \cite{Juice}, and it is composed from the scores of $61$ consumers and the average score for $15$ attributes given by $9$ panelists. This data is available at "\url{http://www.sensometric.org/datasets}".
Consumers evaluated $10$ muscadine grape juices for overall impression, appearance, aroma, color, and flavor. The name of the $10$ studied muscadine grape cultivars are given in the Table $\ref{tabname2}$. These $10$  juices are examined for aroma, basic tastes, aromatics, feeling factors by the group of Sensory Spectrum trained panelists. Likewise to white corn tortilla chips data set, this data set is composed from two tables: consumers notes table and panelists notes table. The first table contains the consumers evaluation of overall impression, appearance, aroma, color and flavor on the 9-point hedonic  scale and the second one contains the average score for each attribute after evaluation of the $9$ panelists for the basic tastes, aromatics and feeling factors attributes  for all muscadine juices.

\noindent The target matrix here is a matrix of dimension $10 \times 15$ constructed by the average score given by the panelists and the reference matrix is a matrix of dimension $10 \times 10$ constructed by the Euclidean distances between the consumers scores for the different cultivars of muscadine grape juices. To quote some of the studied attributes: sweet, sour,	cooked muscadine, cooked grape,	musty,	green unripe, floral	apple/pear,	fermented \dots

 \begin{table}[h]
\centering
\begin{minipage}[t]{.48\linewidth}
\begin{tabular}{|l|c|} 
\hline  
\textbf{Tortilla Chip names} & abb.\\
\hline
Best Yet White Corn & BYW \\
Green Mountain Gringo& GMG \\
Guy's Restaurant Rounds& GUY \\
Medallion White Corn& MED \\
Mission Strips & MIS \\
Mission Triangle & MIT \\
Oak Creek Farms-White Corn & OAK \\
Santita's & SAN \\
Tostito's Bite Size & TOB\\
Tom's White Corn & TOM \\
Tostito's Restaurant Style & TOR  \\
\hline \end{tabular}
\caption{White corn tortilla chip product names and labels}\label{tabname1}
\end{minipage}
\hfill 
\begin{minipage}[t]{.48\linewidth}
\begin{tabular}{|l|c|cc|} 
\hline
\textbf{Muscadine juice names} & abb.\\
\hline
Black Beauty & BB\\
Carlos& CA\\
Granny Val& GV\\
Ison&  IS\\
Nestitt& ME\\
Commercial Red& CR\\
Commercial White& CW\\
Southern Home& SH\\
Summit& SUM\\
Supreme& SUP\\
&\\
\hline
\end{tabular}
\caption{Muscadine grape juice names and labels}\label{tabname2}
\end{minipage}
\end{table}

\subsection{Experimental setup}
Random model of multidimensional fitting method is applied in the  independent and dependent cases of the components of vectors $\varepsilon_i$ for $i=1,\dots,n$. The presence of  Laguerre polynomial $\mathcal{L}_{\frac{1}{2}}^{\frac{p}{2}-1}\left(-\frac{\lambda_{ij}^2}{2}\right)$ in the objective function of problem (P$_1$) complicates the optimization and  makes the computation time too long.  A way to simplify the optimization resolution  is to calculate before the optimization a large set of Laguerre polynomial values corresponding to a large possible values of $\lambda^2_{ij}$ as the Laguerre polynomial presented in the expectation of the error $\Delta$ is related to the value of $\lambda_{ij}^2$. So we define $Z$ a set in $\mathbb{R}$ which contains many possible values of $\lambda_{ij}^2$ that can be used during the optimization. For all $1 \leq i<j \leq n$, the value of $\lambda_{ij}^2$ is proportional to the distance $\Vert X_i+\theta_i-X_j-\theta_j\Vert_{_2}^2$ thus the set $Z$ is related to the data set by the value of $d^2_{ij}$ for all $i,j$ as the objective of our optimization problem is to approach $\Vert X_i+\theta_i-X_j-\theta_j\Vert_{_2}$ to $d_{ij}$. Therefore, we define  $Z$ as follows: $Z=[-\overline{d^2}/(4 \sigma^2)+\ell,0]$ where $\overline{d^2}$ is the mean of squared distances $d_{ij}$ and $\ell$ is a value which gives the length of $Z$ and related to the maximal value of $d^2_{ij}$ in order to cover the largest possible values of $\lambda^2_{ij}$. The increment between the elements of $Z$ is taken equal to $10^{-2}$. After that, during optimization,  each value of $\lambda_{ij}^2$  calculated with a particular $\theta_i$ and $\theta_j$ is replaced by the nearest value  in $Z$ and the Laguerre polynomial value associated to this value is injected directly in the objective function.  This simplification gives results close to the results obtained directly by optimizing (P$_1$) and reduce thousandth times the resolution time. 

\noindent Moreover, the choice of $\sigma$ and the scale parameter $a$ are crucial to obtain good results. Therefore, the value of $\sigma$ is taken  equal to the mean of the $p$ standard deviation calculated on the target matrix $X$, so we can write $\displaystyle \sigma=\frac{\sum_{k=1}^p \sigma_k}{p}$. Concerning the parameter $a$, we calculate it using the following algorithm:
\begin{algorithm}[H] 
\caption{}
\begin{algorithmic}\label{algoa}
\STATE{Initialization: $\Theta^0=(\theta^0_1|\dots|\theta^0_n)=(\textbf{0}| \dots| \textbf{0})$.}
\FOR{$t=1$ to $N_2$} 
\STATE {Solve problem $(P_a)$:\\ $\displaystyle \underset{a \in \mathbb{R}^+}{\min} \sum_{1\leq i<j\leq n}(d_{ij}-a\Vert X_i+\theta^{t-1}_i-X_j-\theta^{t-1}_j\Vert)^2  $}.
\STATE{Solve problem  $(P_\theta)$:\\ $\displaystyle \underset{(\theta_1,\dots,\theta_n) \in \mathbb{R}^p}{\min} \sum_{1\leq i<j\leq n}(d_{ij}-a^t\Vert X_i+\theta_i-X_j-\theta_j\Vert)^2  $}.
\STATE{$\Theta^t=(\theta^t_1,\dots,\theta^t_n)$ solution of problem $(P_\theta)$}.
\ENDFOR
\end{algorithmic}
\end{algorithm}

\noindent  NLopt library (Version 2.4.2) \cite{Nlopt} implanted in language $C$ a free and open-source library is used to solve problem (P${_1}$) as this problem is a non-linear and non-convex optimization problem. In this library, numerous algorithms exist to solve such non-linear optimization problems. In our application, we choose Sbplx algorithm which is based on Subplex method \cite{Subplex} that is a generalization of Nelder-mead simplex.

\noindent Concerning simulation algorithm,  the covariance matrix $\Sigma$  is given by the covariance of matrix $X$ multiply by a constant $c$. The parameter  $\xi$ presented in the proposal distribution is taken equal to $10^{-4}$.  Concerning the temperature parameter, we take it equal to $T=100$.
Moreover, during simulation, the number of iterations is taken equal to $N_1=300$.

\subsection{Results}
\subsubsection{Optimization results}
First, we want to define the different values of parameters $a$, $\sigma$ and  $\ell$ for the two data sets. Table $\ref{taba}$ gives these values:
\begin{table}[htpb]
\begin{center}
\begin{tabular}{|c|ccc|}
  \hline
 & $a$ & $\sigma$ & $\ell$\\
 \hline
 White corn tortilla chips & $ 26.37$ & $0.55$ & $1000 $\\
 Muscadine grape juices & $ 36.9$ & $0.64$ & $1000$  \\
  \hline
\end{tabular} 
\caption{The values of parameters $a$, $\sigma$ and  $\ell$ for the two data sets.}\label{taba}
\end{center}
\end{table} 

\noindent The values of $a$ for the two data sets are calculated using algorithm $\ref{algoa}$. Figure $\ref{figa}$ depicts the trace plots of the value of $a$ at each iteration for the two data sets. We show clearly that the value of $a$ converge to an optimal value.
Concerning the value of $\ell$, a choice of $1000$ for the two data set can be reasonable  as the maximum value of the squared distances in the two data sets is in the range of $1000$.
\begin{figure}[htpb] 
\centerfloat 
 \begin{minipage}[t]{.65\linewidth}
            \makebox[\textwidth]{  \includegraphics[width=8cm, height=5.33cm]{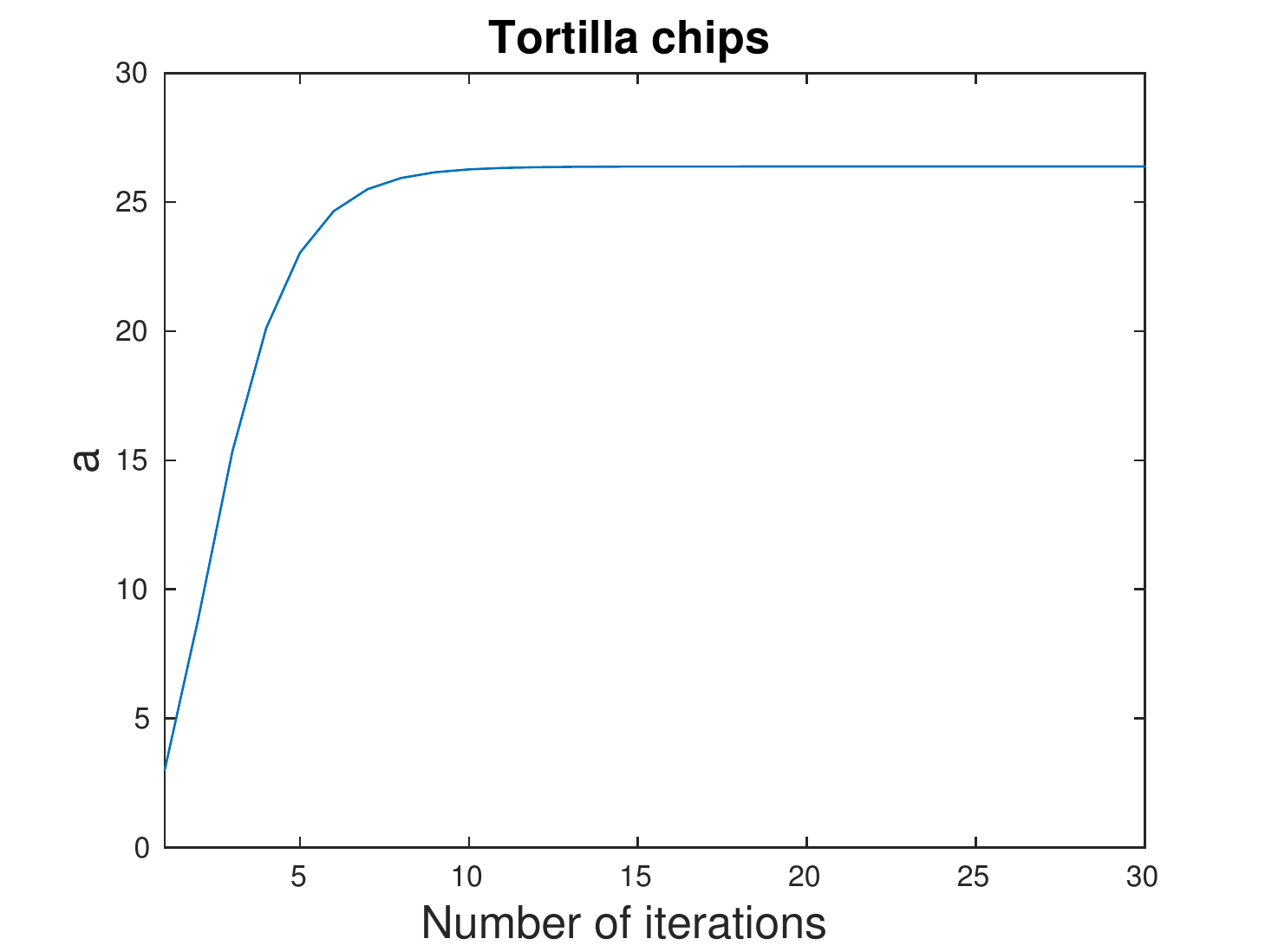}}
              
    \end{minipage}    
\hfill
\begin{minipage}[t]{.65\linewidth}
             \makebox[\textwidth]{ \includegraphics[width=8cm, height=5.33cm]{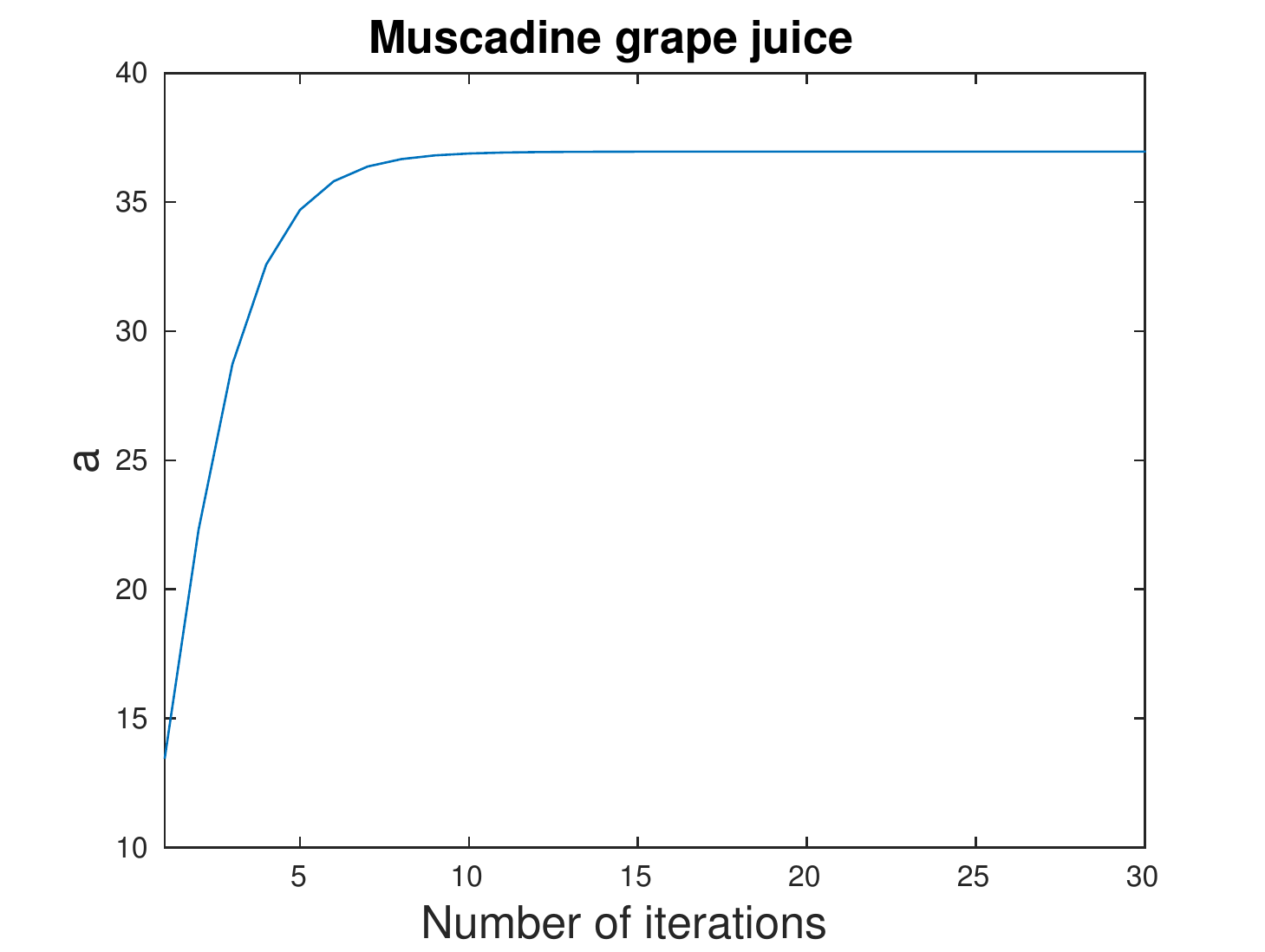}   }  
        \end{minipage}  
        \caption{The trace plots of the results of the algorithm $\ref{algoa}$ for tortilla chips and muscadine grape juice data sets.}\label{figa}
\end{figure} 

\noindent After parameters determination, the statistical test developed in Section \ref{Sectest}  has been applied to perform the interest of the displacements of the points. The values of the ratio $\mathcal{R}=\displaystyle \frac{\mathbb{V}ar(\Delta)}{(\Delta_0-\mathbb{E}(\Delta))^2}$  calculated for the two data sets are given in the Table $\ref{tabSig}$.
\begin{table}[htpb]
\begin{center}
\begin{tabular}{|c|c|}
  \hline
 & $\mathcal{R}$ \\
 \hline
 White corn tortilla chips & $0.0239$ \\
 Muscadine grape juices & $ 0.0153$\\
  \hline
\end{tabular} 
\caption{The values of $\mathcal{R}$ for white corn tortilla chips and muscadine grape juices data sets.}\label{tabSig}
\end{center}
\end{table} 
\noindent As the two values of $\mathcal{R}$ for the two data sets are smaller than $0.05$,  so the statistical test is significant for $\alpha=0.05$ therefore we reject the null hypothesis $(\mathcal{H}_0)$ and we accept the alternative  hypothesis $(\mathcal{H}_1)$. Thus, the movements of points $i$ and $j$ through vectors $\theta_i$ and $\theta_j$ are necessary to approach the distances $\Vert X_i+\theta_i-X_j-\theta_j\Vert$ to $d_{ij}$ for all $1\leq i<j\leq n$. After the statistical test, problem (P$_1$) has been solved using different values of regularization parameter $\eta$. 

\noindent Tables $\ref{tabres1}$ and $\ref{tabres2}$ show the different  values of the expectation of error $\Delta$ and the number of non-null displacements $\theta_{ik}$ after optimization for different values of $\eta$ obtained after optimization. We remark that when $\eta$ increases, the number of displacements decreases and when $\eta$ becomes too large, the number of displacements tends to zero and nothing moves.

\begin{table}[h]
\centering
\begin{minipage}[t]{.4\linewidth}
\caption*{\textbf{White corn tortilla chips}}
\begin{tabular}{|c|cc|} 
    \hline
      $\eta$&$\mathbb{E}(\Delta)$&\#($\theta_{ik} \neq 0)$ \\
    \cline{1-3} 
         $0$ &$438896$ & $407 $	\\
                       $10$ &$438892$&$ 405$\\
        $10^2$&$ 444855$  & $401 $\\
   $10^3$&$393359$ &$397 $\\     
               $10^4$&$447001$ &$360 $\\     
                	$	10^5$& 	$709809$&$ 233$\\
                    		$	2\times 10^5$& 	$1450193$&$ 190$\\
                    		$	4\times 10^5$& 	$4558334$&$ 153$\\
                    		$	6\times 10^5$& 	$8229189$&$ 121$\\
                    		$	7\times 10^5$& 	$12330002$&$ 109$\\
                        	$	10^6$& 	$19927843$ &$87 $\\
                        	$	 10^7$&$229240376$  &$ 0$\\
    \hline   
                
\end{tabular}
\caption{The values of $\mathbb{E}(\Delta)$ and the number of non-null displacements for different values of $\eta$ for  tortilla chips data set.}\label{tabres1}
\end{minipage}
\hfill 
\begin{minipage}[t]{.4\linewidth}
\caption*{\textbf{Muscadine grape juices}}
\begin{tabular}{|c|cc|} 
    \hline
      $\eta$&$\mathbb{E}(\Delta)$&\#($\theta_{ik} \neq 0)$ \\
    \cline{1-3} 
         $0$ &$127845$ & $150 $	\\
                       $10$ &$127845$&$ 149$\\
      $ 10^2$&$127845$ &$147$\\
        $ 10^3$&$127952$ &$133$\\
                      $ 4\times 10^3$&$128799$ &$128$\\
                   $ 6\times 10^3$&$130641$ &$116$\\
                           $ 8\times 10^3$&$140656$ &$98$\\
                          $10^4$&$168285$ &$86 $\\  
                           $2 \times 10^4$&$276341$ &$54 $\\  
                           $4 \times 10^4$&$551644$ &$28 $\\     
                        	$ 10^5$& 	$1336073$ &$10 $\\
                        	$	 10^6$&$ 3517594$  &$ 0$\\
                      
    \hline
\end{tabular}
\caption{The values of $\mathbb{E}(\Delta)$ and the number of non-null displacements for different values of $\eta$ for  muscadine juices data set.}\label{tabres2}
\end{minipage}
\end{table}

\noindent A way to choose  the value of $\eta$ is to determine the number of misplaced points which must be moved to fit the distances. To find these misplaced points, we use the criterion of selection points $\rho_i$ presented in Section $\ref{sec1}$.  

\noindent Table $\ref{tabRho}$ shows the values of this criterion. We have seen that for a fixed real number $\varrho$  between $0$ and $1$, if $\rho_i>\varrho$ we consider $i$ as misplaced point. So, by taking $\varrho=0.1$ for the two data sets,  we can detect $3$  misplaced points for white corn tortilla chips and $4$ for muscadine grape juices which is equivalent to $3\times 37=111$ values of $\theta_{ik} \neq 0$ for tortilla chips and $4 \times 15=60$ for muscadine juices. Then, by referring to Tables $\ref{tabres1}$ and $\ref{tabres2}$, we  choose the value of $\eta$ that gives a number of displacement  close to that obtained using the criterion $\rho_i$ for tortilla chips and muscadine juices. Indeed, a value of $\eta$ equal to $7\times 10^5$  gives a number of displacements equal to $109$ displacements that is close to $111$, so we can take $\eta=7 \times 10^5$. Similarly, we choose $\eta= 2\times 10^4$ for muscadine juices data set.
Noted that by changing the value of $\varrho$, we can detect more misplaced points so this choice must be reasonable.

\begin{table}[h]
\begin{center}
\resizebox{\textwidth}{!}{\begin{tabular}{|c||c|ccccccccccc|}
  \hline
 $i$&& $1$& $2$& $3$& $4$ & $5$& $6$& $7$& $8$& $9$& $10$ &$11$\\
 \hline 
$\rho_i$& $\mathcal{D}_1$&    $ 0.0640$ &   $\textbf{0.1253}$ &   $0.0706$&   $ 0.0700 $ &  $0.0931$   & $0.0865$ &   $\textbf{0.1270}$&    $0.0811$ &$0.0828$ &   $\textbf{0.1088}$  & $0.0908$\\
    \cline{2-13}
&$\mathcal{D}_2$&$0.0960$  &  $0.0589$ &   $0.0601$ &   $0.0900$ &   $0.0966 $  & $\textbf{0.1483} $  & $0.0954$ &   $\textbf{0.1028}$ &  $ \textbf{0.132}$   & $\textbf{0.1197}$&\\
\hline
\end{tabular} }
\caption{The values of criterion $\rho_i$ for the $11$ white corn tortilla chips samples $(\mathcal{D}_1)$ and the $10$ muscadine grape juices $(\mathcal{D}_2)$. The bold values corresponds to the values where $\rho_i> 0.1$.}\label{tabRho}
\end{center}
\end{table}

\noindent Besides, if the number of desirable displacements is fixed by the user then it is not needed to compute the ratio $\rho$ and in the same way we can choose the value of $\eta$. So, the choice of $\eta$ is always related to the objective which is aimed at.

The objective of the study is to determine the acceptable attributes categories of white corn tortilla chips and muscadine grape juices. Using our method we want to determine the product characteristics that must be changed to match with the consumers preference.

\begin{table}[H]
\begin{center}
\resizebox{\textwidth}{!}{\begin{tabular}{|c|ccccccccccc|}
\hline
& & & &\textbf{White}&  \textbf{Corn }& \textbf{Tortilla}& \textbf{Chips} & & & & \\
\cline{2-12}
 & BYW &    GMG &    GUY &    MED &    MIS &    MIT &    OAK &    SAN &    TOB &    TOM & TOR \\
  \cline{2-12}
& & & & & \textbf{Flavor} & & & & & &\\
    \hline
Sweet&	$0$&	$0$&	$0$&	$1.2451$&	$0$&	$0$&	$-0.9751$&	$0$&	$1.6118$&	$1.3805$&	$0$\\
Salt&	$0$&	$1.6930$&	$0$&	$1.9880$&	$0$&	$0$&	$1.5206$&	$0$&	$0$&	$0$&	$0$\\
sour&	$0$&	$0$&	$0$&	$0$&	$0$&	$0$&	$-2.7781$&	$0.7170$&	$0$&	$0$&	$0$\\
Astringent&	$0$&	$1.8238$&	$0$&	$0$&	$-1.3989$&	$0$&	$0$&	$0$&	$0$&	$-0.6097$&	$1.8035$\\
Grain complex&	$0$&	$-1.3154$&	$1.5591$&	$2.3519$&	$0$&	$0$&	$0$&	$0$&	$0$&	$0$&	$0$\\
Raw corn&	$0$&	$-1.9829$&	$0$&	$0$&	$0$&	$0$&	$-2.8642$&	$0$&	$0$&	$0$&	$0$\\
Masa&	$0$&	$0$&	$0$&	$0$&	$1.4716$&	$0$&	$-1.6456$&	$0$&	$0$&	$-1.4004$&	$0$\\
Toasted grain&	$0$&	$-1.2323$&	$-1.3444$&	$0$&	$0$&	$0$&	$-0.4043$&	$1.4624$&	$1.0626$&	$0$&	$-1.5607$\\
Heated oil&	$0$&	$0$&	$0$&	$0$&	$0$&	$0$&	$2.0665$&	$0$&	$0$&	$-1.6899$&	$0$\\
Scorched&	$0$&	$0$&	$0$&	$0$&	$-3.1855$&	$0$&	$0$&	$0$&	$0$&	$0$&	$0$\\
Cardboard&	$0.3165$&	$1.3200$&	$0$&	$-1.9339$&	$0$&	$0$&	$0$&	$1.3330$&	$1.6492$&	$0$&	$0$\\
\hline
\hline  
     \cline{2-12}
& & & & & \textbf{Texture} & & & & & &\\
\hline
Oily/ greasy lip&	$-1.1890$&	$1.4904$&	$0$&	$0$&	$0$&	$0$&	$0$&	$0$&	$0$&	$-1.5958$&	$0$\\
Loose particles&	$-1.8266$&	$1.9488$&	$0$&	$0$&	$0$&	$0$&	$0$&	$0$&	$0$&	$0$&	$0$\\
Hardness&	$0$&	$-2.3640$&	$0$&	$0$&	$0$&	$0$&	$0$&	$0$&	$0$&	$0$&	$0$\\
Crispness&	$0$&	$1.3160$&	$0$&	$0$&	$0$&	$0$&	$1.2906$&	$0$&	$-1.2140$&	$0$&	$-1.8826$\\
Cohesiveness of mass&	$1.6169$&	$-1.2832$&	$0$&	$0$&	$0$&	$0$&	$-1.7938$&	$0$&	$0$&	$0$&	$-1.5653$\\
Roughness of mass&	$-1.2758$&	$1.3458$&	$0$&	$0$&	$0$&	$0$&	$0$&	$0$&	$0$&	$-1.8818$&	$0$\\
Moistness of mass&	$0$&	$-1.2495$&	$0$&	$0$&	$0$&	$1.7023$&	$-1.8192$&	$0$&	$0$&	$0$&	$0$\\
Moisture absorption&	$0.9027$&	$0$&	$0$&	$0$&	$-1.4304$&	$-1.7818$&	$0$&	$0$&	$0$&	$1.0449$&	$0$\\
Persistence of crisp&	$-2.2752$&	$0$&	$0$&	$0$&	$0$&	$0$&	$0.1331$&	$0$&	$-1.3031$&	$0$&	$0$\\
Toothpack&	$0.1483$&	$0$&	$0$&	$0$&	$0$&	$2.6819$&	$0$&	$0$&	$-0.5108$&	$0$&	$0$\\
 \hline
\hline  
     \cline{2-12}
& & & & & \textbf{Appearance} & & & & & &\\
\hline
Degree of Whitenes&	$0$&	$2.5915$&	$0$&	$0$&	$0$&	$1.9849$&	$0$&	$0$&	$0$&	$0$&	$0$\\
Grain Flecks&	$-0.8481$&	$1.7684$&	$0$&	$0$&	$0$&	$1.7982$&	$0$&	$0$&	$0$&	$0$&	$0$\\
Char Marks&	$0$&	$-1.1166$&	$0$&	$0$&	$0$&	$0$&	$0$&	$0$&	$1.5467$&	$0$&	$1.5647$\\
Micro Surface Particles&	$0$&	$1.6609$&	$0$&	$0$&	$0$&	$0$&	$0$&	$0$&	$0$&	$-1.9241$&	$0$\\
Amount of Bubbles&	$0$&	$0$&	$0$&	$0$&	$0$&	$0$&	$0$&	$0$&	$-1.7983$&	$0$&	$1.0500$\\
\hline
\end{tabular}}
\caption{The values of displacements $\theta_{ik}$ where $i=1,\dots,11$ is the corn tortilla chip sample and $k$ is the attributes of flavor, texture, appearance categories. Only the detected descriptive attributes are given in this table.}\label{tabTortilla}
\end{center}
\end{table}

\begin{table}[H]
\begin{center}
\resizebox{\textwidth}{!}{\begin{tabular}{|c|cccccccccc|}
\hline
& & & &\textbf{Muscadine }&  \textbf{grape }& \textbf{juices}&  & && \\
\cline{2-11}
& BB&	CA&	GV&	IS&	ME&	CR&	CW&	SH&	SUM&	SUP\\
\cline{2-11}
& & & & & \textbf{Basic tastes} & & & & &\\
    \hline
Sweet& $-1.1850$&	$0.5018$&	$0.6948$&	$0$&	$0$&	$0$&	$0$&	$0$&	$0$&	$0$\\
Sour& $-1.0612$&	$0.2443$&	$0$&	$0$&	$0$&	$-1.0543$&	$0$&	$0$&	$0$&	$0$\\

\hline
\hline  
     \cline{2-11}
& & & & & \textbf{Aromatics} & & & & & \\
\hline
Cooked muscadine&	$-0.1948$&	$0$&	$0$&	$0$&	$0$&	$-0.4885$&	$0.4100$&	$0.5171$&	$0$&	$0$\\
Cooked grape&	$0$&	$-0.5705$&	$0$&	$0$&	$0$&	$-0.6065$&	$0$&	$0$&	$-0.8350$&	$0$\\
Musty&	$0$&	$-0.6460$&	$0$&	$-0.4153$&	$0.7202$&	$0$&	$0.0497$&	$0$&	$0$&	$0$\\
Green/unripe&	$0$&	$0.3416$&	$0$&	$0$&	$0$&	$0$&	$0$&	$-0.3443$&	$0$&	$0$\\
Floral&	$0$&	$0.4433$&	$0.4345$&	$0$&	$-0.9344$&	$0$&	$-0.5002$&	$0.4070$&	$-0.9426$&	$0$\\
Apple/Pear&	$-1.0636$&	$0$&	$0.5652$&	$0$&	$0$&	$-0.4278$&	$0$&	$0$&	$0$&	$0.6464$\\
Fermented&	$-1.2588$&	$0$&	$0$&	$0$&	$0$&	$-0.7521$&	$0$&	$0$&	$0.3882$&	$0$\\
Metallic&	$0$&	$0.6494$&	$0$&	$-1.1046$&	$0.6358$&	$0.6547$&	$0$&	$0.6520$&	$-0.7677$&	$0$\\
 \hline
\hline  
     \cline{2-11}
& & & & & \textbf{Feeling factors} & & & & & \\
\hline
Astringent& $0.3984$&         $0$&         $0 $&  $-0.8488$        & $0$ &   $0.3228$&    $0.7118$&         $0$&  $0$ &   $0.4081$\\
\hline
\end{tabular}}
\caption{The values of displacements $\theta_{ik}$ where $i=1,\dots,10$ is the $10$ muscadine grape juices and $k$ is the attributes of basic tastes, aromatics, feeling factors categories. Only the detected descriptive attributes are given in this table.} \label{tabRaisin}
\end{center}
\end{table}

As we have seen we are interested in our method to fit the characteristics of product to the consumer acceptance rates, so null displacements can be interpreted as consumers satisfaction. Globally, for each categories of attributes, we can calculate the proportion of the null displacements. This proportion is given by: 
\[p_{_\text{C}}=\frac{\text{number of } (\theta_{ik}=0) \text{ in category } C_{_\mathcal{D}} }{\text{total number of } \theta_{ik} \text{ in category } C_{_{\mathcal{D}}}  }\]
where $\mathcal{D}=\{\mathcal{D}_1, \mathcal{D}_2\}$, C$_{_{\mathcal{D}_1}}\in\{\text{Flavor, Texture, Appearance}\}$ is the category of white corn torilla chips $(\mathcal{D}_1)$ and C$_{_{\mathcal{D}_2}}\in\{ \text{ Basic tastes, Aromatics, Fellings factors }\}$ is the category of muscadine grape juices $(\mathcal{D}_2)$. For each data set, the proportion $p_{_{\text{C}}}$ is calculated from Table $\ref{tabTortilla}$ or Table $\ref{tabRaisin}$.   
\begin{table}[H]
\centering
\begin{minipage}[t]{.48\linewidth}
\begin{tabular}{|c|c|}
\hline
Tortilla attributes & \\ Categories & $P_{_C}$\\
  \hline
Flavor & $0.71 $\\
Texture &    $0.72$\\
Appearance &    $0.78$\\
\hline
\end{tabular} 
\end{minipage}
\hfill 
\begin{minipage}[t]{.48\linewidth}
\begin{tabular}{|c|c|}
\hline
Muscadine attributes&\\ Categories & $P_{_C}$\\
  \hline
Basic tastes & $0.70$\\
Aromatics &    $0.60$\\
Feeling factors &    $0.50$\\
\hline
\end{tabular}
\end{minipage}
\caption{The proportion values for different attributes categories for white corn tortilla chips and muscadine grape juices.}\label{tabprop}
\end{table}

\noindent The proportion values for different categories for the two data sets are given in the Table $\ref{tabprop}$. For white corn tortilla chips, the flavor, texture and appearance attributes categories have approximatively  the same proportion values that it is equal in average to $0.75$. So, $25\%$ of the characteristics products must be moved to make the products characteristics as acceptable as possible by the consumers. Thus, we can conclude that the overall characteristics of tortilla chips are well accepted by the consumers.  Concerning muscadine grape juices, we notice that the basic tastes attributes are the most acceptable attributes among the two other attributes categories as just $30\%$ of the attributes must be changed to fit the consumer scores whereas, $40\%$ and $50\%$ of the aromatics and feeling factors attributes categories  must respectively be changed to make these characteristics acceptable by the consumers.  

\subsubsection{Simulation results}
Algorithm $1$ has been applied to the two data sets. 
The constant $c$ multiplied by covariance matrix of data is taken equal to $10^{-3}$ for the two data sets.

\noindent Figure $\ref{figSimAlgo1}$ shows that the minimal value of  error $\Delta$ obtained by simulation is equal to $18129$ after $150$ iterations for the white corn tortilla chips and $4838$ after $50 $ iterations for the muscadine grape juices. The results of displacements for these two data sets are given in Figure $\ref{figDepAlgo1}$ and $\ref{figDepAlgo1Rai}$. In these figures, we compare the displacements obtained by simulation with those obtained by optimization of problem P$_0$ (without penalization term) as the parsimonious choice of displacements is not taken into account in the simulation algorithm.

\begin{figure}[!ht]
\centerfloat 
      \begin{subfigure}[b]{8cm}
         \makebox[\textwidth]{\includegraphics[width=8cm,height=5.33cm]{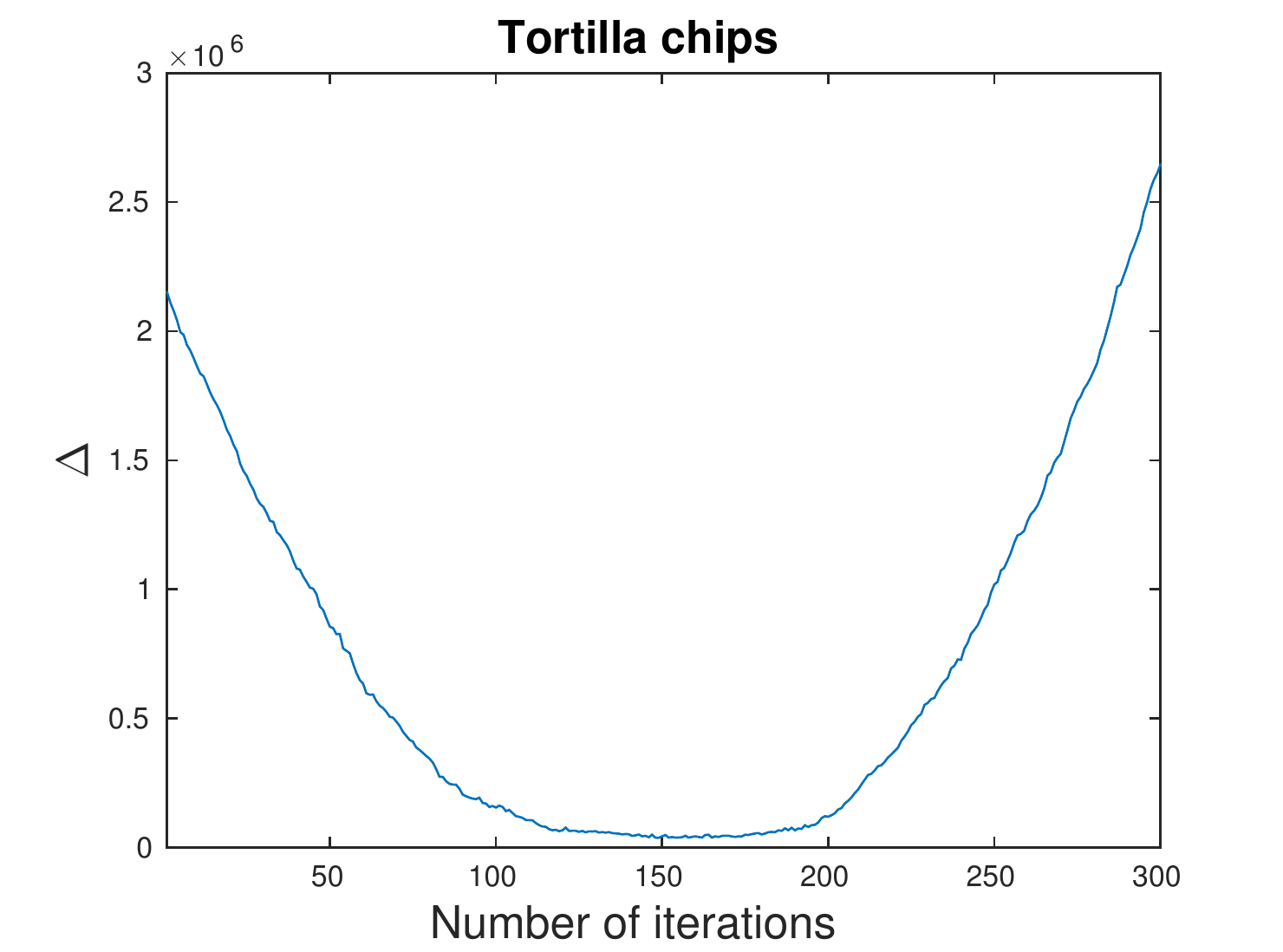}}
     \end{subfigure}
    \begin{subfigure}[b]{8cm}
     \makebox[\textwidth]{    \includegraphics[width=8cm,height=5.33cm]{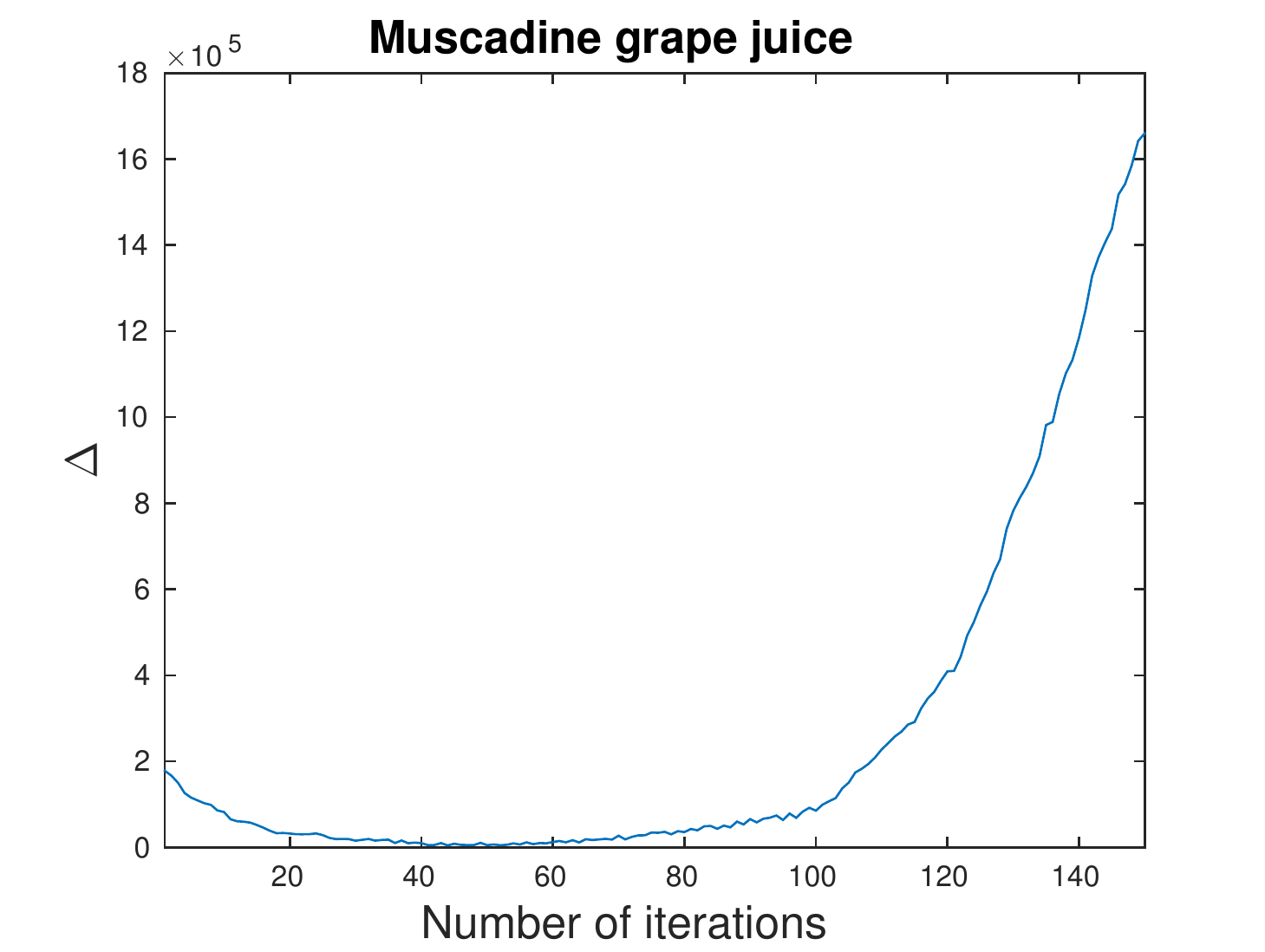}}
        \end{subfigure}\\
    \caption{trace plot of the error $\Delta$ using algorithm $1$ for white corn tortilla chips and muscadine grape juices data sets. }\label{figSimAlgo1}
\end{figure}

This comparison between optimization and simulation results indicates that using simulation technique we have succeeded in finding similar displacements with a value of $\Delta$ smaller than the value of the expectation calculated in the independent case. What is interesting here that the displacement obtained after simulation is not very different for most of the points.  The important displacements obtained in the optimization and simulation results are close. So by taking some threshold to detect the important displacements, we can detect the same important displacements in two different ways.  

\begin{figure}[!h]
    \centering
    \begin{subfigure}[b]{10cm}
       \makebox[\textwidth]{  \includegraphics[width=16cm, height=8cm]{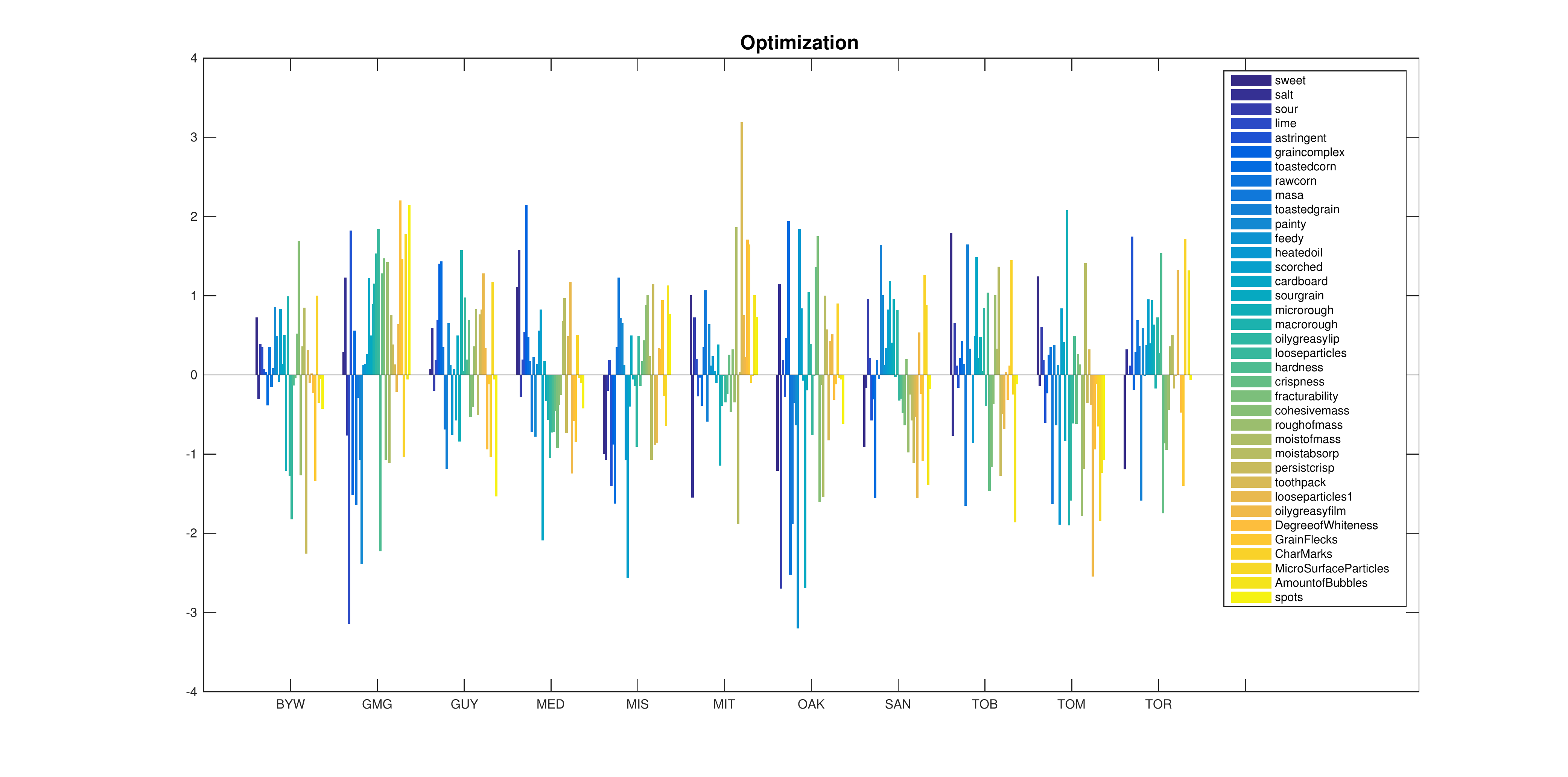}}
    \end{subfigure}
    \begin{subfigure}[b]{10cm}
     \makebox[\textwidth]{    \includegraphics[width=16cm,height=8cm]{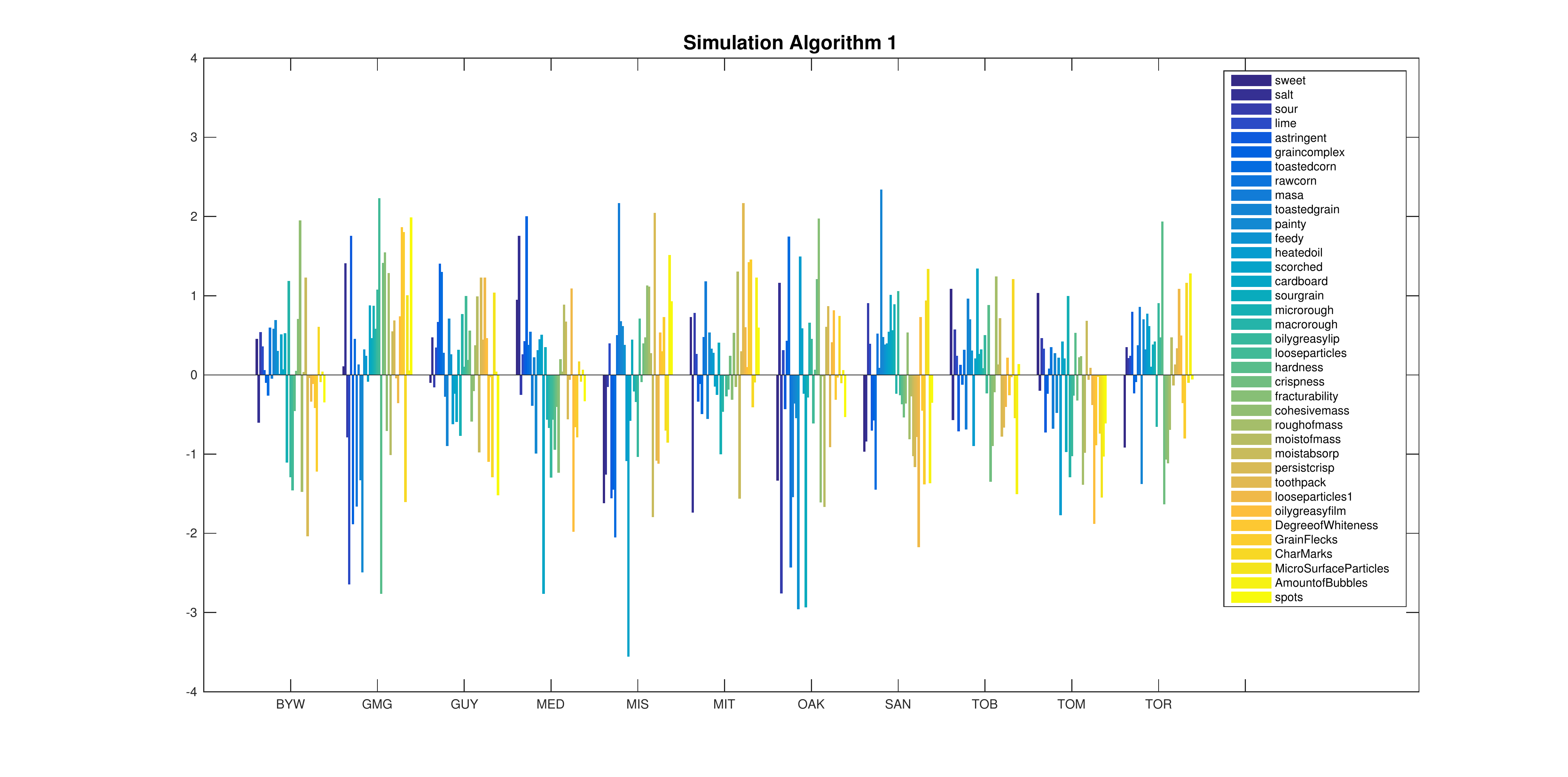}}        
    \end{subfigure} 
\caption{The displacements for the different attributes of the $11$ tortilla chip samples obtained by optimization and simulation. }\label{figDepAlgo1}
\end{figure}

\begin{figure}[!h]
    \centering
    \begin{subfigure}[b]{10cm}
       \makebox[\textwidth]{  \includegraphics[width=16cm, height=8cm]{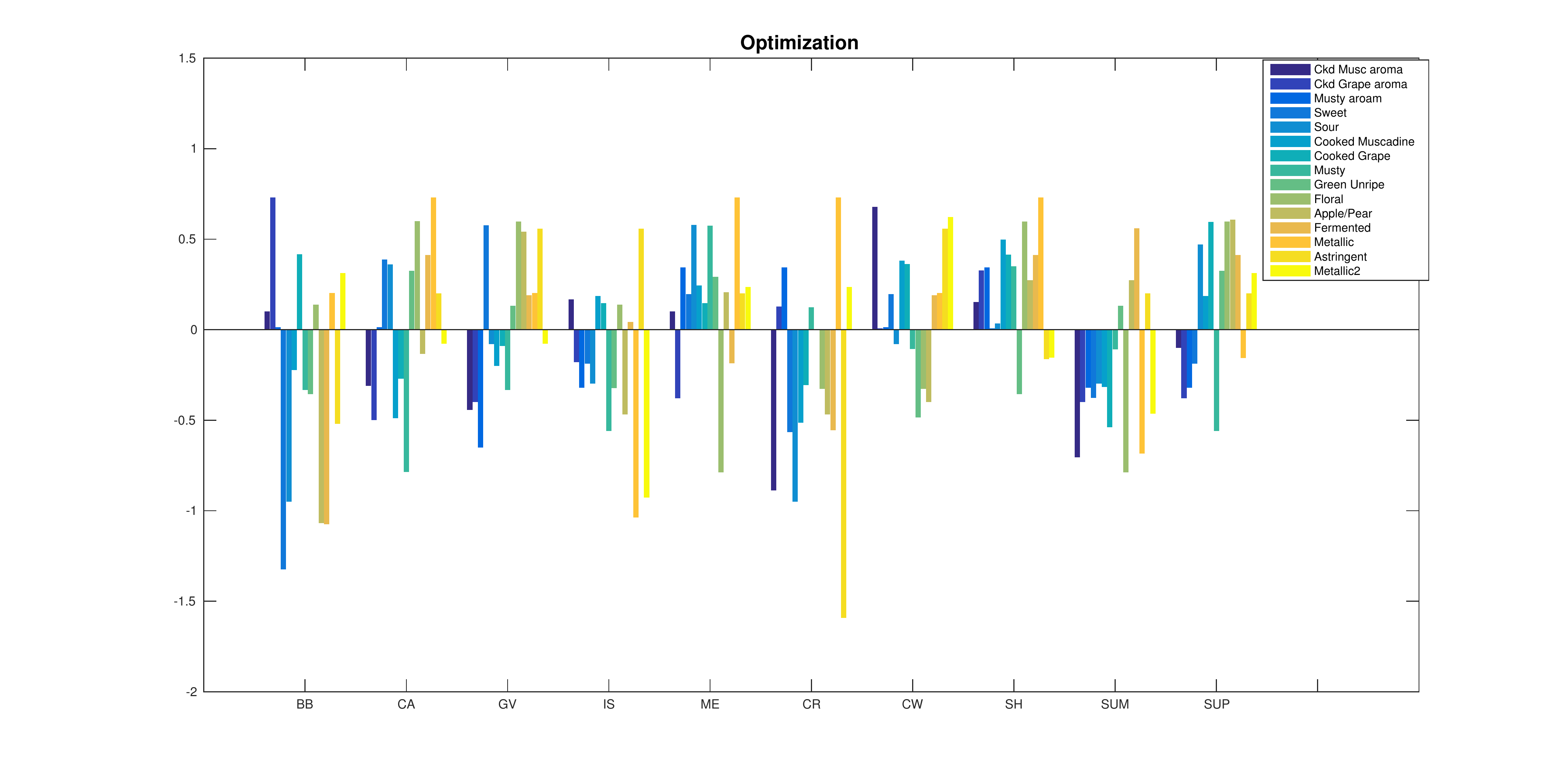}}
    \end{subfigure}
    \begin{subfigure}[b]{10cm}
     \makebox[\textwidth]{    \includegraphics[width=16cm,height=8cm]{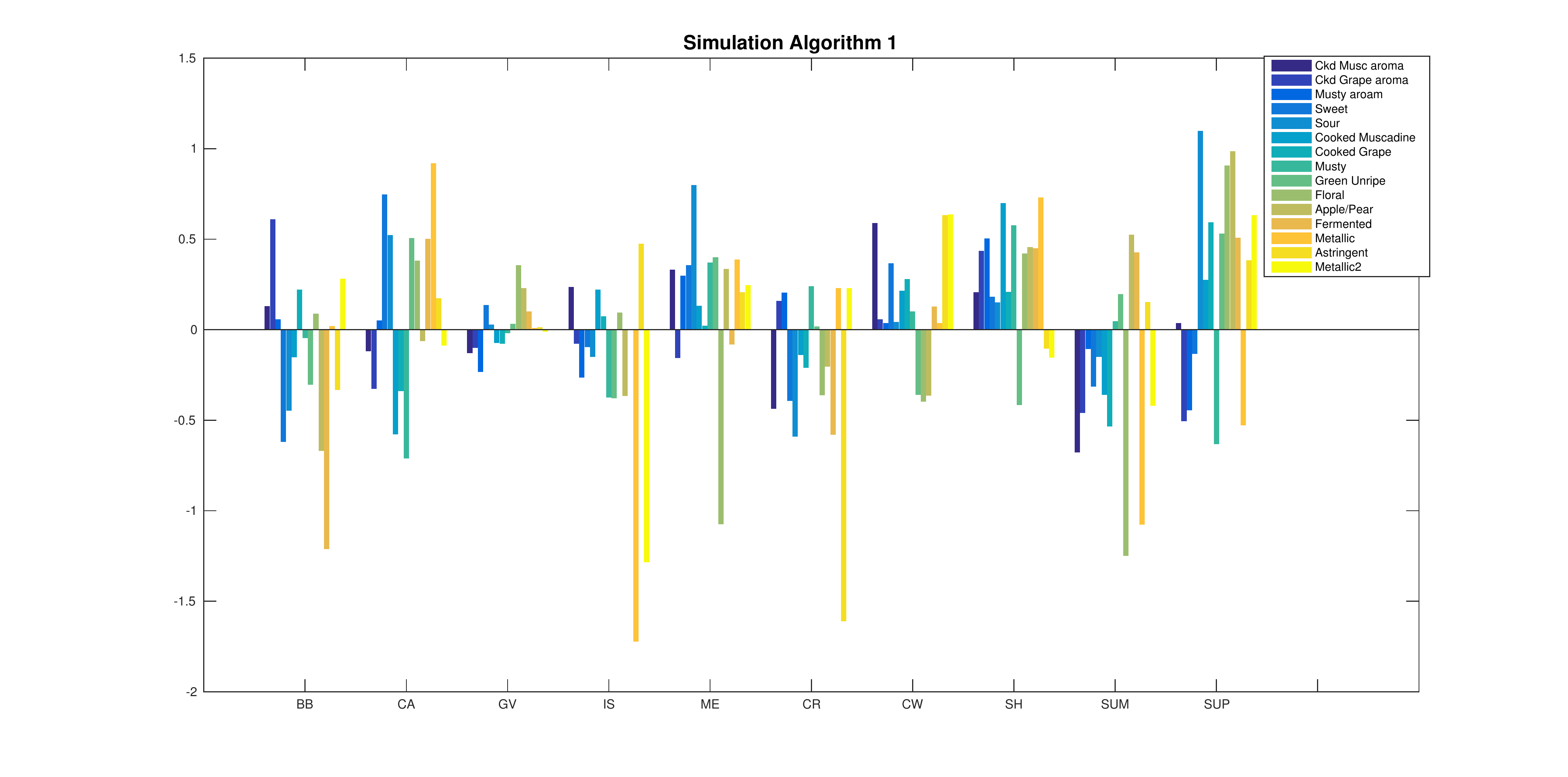}}        
    \end{subfigure}
\caption{The displacements for the different attributes of the $10$ muscadine grape juices obtained using optimization and simulation.}\label{figDepAlgo1Rai}
\end{figure}

\subsection{Discussion}
Several papers in  food quality and preference domains study the relation between consumers preference and the characteristics of products in order to find the must acceptable characteristics of these products by the consumers \cite{SS1}. Preference mapping techniques can be applied using just the consumers rates for each product, we address 'internal preference mapping', or by taking an additionally data describing the products with a series of criteria, we address 'external preference mapping' \cite{SS2}. As our method based on the fitting of two matrices so the comparative with external preference mapping should be more explicative. The main objective to external preference mapping is to fit the individual consumer rates to the products configuration  by using one of the different regression models among which the quadratic surface model is popular \cite{SS3}. Meullenet in his article \cite{Tortila2} indicates that the preference mapping is determined by determining a  partial least squares (PLS) regression model and the application of Jackknife optimization and this model is used to predict the consumer attributes acceptance. So, the prediction here is related to something subjective whereas our method gives displacements that can be interpretable without introducing subjective effects. Moreover, the displacement of   attributes for each product can be interpretable alone or by taking all the categories of attributes. Thus, using our model we can find a new and simple  methodology to determine the preference mapping of products.

\section{Conclusion}\label{SecConc}
We have presented a new model of multidimensional fitting method by taking into account  random effects occurred. First, the random model of MDF  with a penalized form is presented. Second, a statistical test indicates the significance of the displacements of the points. Then, optimization and simulation algorithms are developed to find these displacements. The application of this method in the sensometrics domain shows the simplest explanation of the sensory  profiles of products according the consumers preference.  Finally, MDF in their deterministic and random model can be also used when the data contains missing data. A 
pretreatment of this data before the application of MDF method to replace these missing values will not impinge the results. Finally, more studies can be developed to adapt MDF method to qualitative and functional data.

\newpage

\appendix
\makeatletter
\def\@seccntformat#1{Appendix~\csname the#1\endcsname:\quad}
\makeatother
\begin{appendices}
As we have seen, the expectation and the variance of the error $\Delta$ are involved in the objective function of problem (P$_1$) and in the statistical test. So, we want to calculate these two quantities for any value of $\theta_1,\dots,\theta_n$.

\section{Five Lemmas}\label{App1}
\noindent We have:
\begin{eqnarray}
e_{ij}&=& \left(d_{ij}-a\lVert X_{i}+\theta_{i}+\varepsilon_{i}-X_{j}-\theta_{j}-\varepsilon_j \rVert_{_2} \right)^2.\nonumber
\end{eqnarray}
\noindent By developing the expression of $e_{ij}$, we obtain: 
\begin{eqnarray}
e_{ij}=d_{{ij}}^2+a^2\lVert X_{i}+\theta_{i}+\varepsilon_{i}-X_{j}-\theta_{j}-\varepsilon_j \rVert_{_2}^2-2\ a \ d_{ij} \lVert X_{i}+\theta_{i}+\varepsilon_{i}-X_{j}-\theta_{j}-\varepsilon_j \rVert_{_2} \label{eqdef1}
\end{eqnarray} 

\noindent We want to present five lemmas that will help us in the calculation of the expectation and variance of $\Delta$.

\begin{lemma}\label{lemma1}
Let $N_{ij}$ be a random variable defined by:
\begin{eqnarray}
 \displaystyle N_{ij}=\sum_{k=1}^{p} \left({\varepsilon_{ik}-\varepsilon_{jk}}\right)^2, \forall 1\leq i<j\leq n \label{eqt1}
\end{eqnarray}
 where the $p$ components of vectors $\varepsilon_i$, for all $i=1,\dots,n $ are  independent and identically normally distributed random variables such that  $\varepsilon_{ik}\rightsquigarrow \mathcal{N}(0,\sigma^2)$ for $k=1,\dots,p$. Then, we have: \[\mathbb{E}(N_{ij})=2\sigma^2 p \text{ et } \mathbb{V}ar(N_{ij})=8\sigma^4 p.\]
\end{lemma}

\begin{proof}
\noindent As, $\varepsilon_{ik}\rightsquigarrow \mathcal{N}(0,\sigma^2)$ and the vectors  $\varepsilon_i$ and $\varepsilon_j$ are independents, we have $\varepsilon_{ik}-\varepsilon_{jk}\rightsquigarrow \mathcal{N}(0,2\sigma^2)$. Thus,
$\displaystyle \sum_{k=1}^{p} \left(\frac{\varepsilon_{ik}-\varepsilon_{jk}}{\sqrt{2}\sigma}\right)^2\rightsquigarrow \chi^2_p$, and  consequently:
\[ \displaystyle \mathbb{E}\left(\sum_{k=1}^{p} \left(\frac{\varepsilon_{ik}-\varepsilon_{jk}}{\sqrt{2}\sigma}\right)^2\right)=p \text{ et } \displaystyle \mathbb{V}ar\left(\sum_{k=1}^{p} \left(\frac{\varepsilon_{ik}-\varepsilon_{jk}}{\sqrt{2}\sigma}\right)^2\right)=2p.\]
So, we obtain:
\begin{eqnarray}
 \mathbb{E}\left(N_{ij} \right)&=&2\sigma^2 p\nonumber\\
  \mathbb{V}ar\left(N_{ij} \right)&=&8\sigma^4 p.\nonumber
 \end{eqnarray}
\end{proof}

\begin{lemma}\label{lemma2}
Let $A_{ij}$ be a random variable defined by: 
\begin{eqnarray}
A_{ij}=\lVert X_{i}+\theta_{i}+\varepsilon_{i}-X_{j}-\theta_{j}-\varepsilon_j \rVert_{_2} \label{eqt2}
\end{eqnarray}
with $\varepsilon_{ik}\rightsquigarrow \mathcal{N}(0,\sigma^2)$. Then, we have:
\begin{eqnarray}
\mathbb{E}(A_{ij})&=&\sqrt{2}\sigma \mu_{ij}\nonumber\\
 \mathbb{V}ar(A_{ij})&=&{2}\sigma^2(p+\lambda_{ij}^{2}-\mu_{ij}^{2})\nonumber\\
\mathbb{E}(A_{ij}^2)&=&{2}\sigma^2 (p+\lambda_{ij}^{2})\nonumber\\ 
\mathbb{V}ar(A_{ij}^2)&=&8\sigma^4(p+2\lambda_{ij}^{2})\nonumber\\ 
 \mathbb{E}\left(A_{ij}^3\right)&=&6\sigma^3 \sqrt{\pi} \mathcal{L}_{\frac{3}{2}}^{\frac{p}{2}-1}\left(-\frac{\lambda_{ij}^{2}}{2}\right)\nonumber\\
  \mathbb{E}\left(A_{ij}^4\right)&=&4\sigma^4(p+\lambda_{ij}^{2})^2+8\sigma^4(p+2\lambda_{ij}^{2})\nonumber
\end{eqnarray}

\noindent where $\displaystyle \mu_{ij}= \sqrt{\frac{\pi}{2}} \mathcal{L}_{\frac{1}{2}}^{\frac{p}{2}-1}\left(-\frac{\lambda_{ij}^{2}}{2}\right)$, $\displaystyle \lambda_{ij}=\sqrt{\sum_{k=1}^p \left(\frac{ x_{ik}+\theta_{ik}-x_{jk}-\theta_{jk}}{\sqrt{2}\sigma}\right)^2}$ and $\mathcal{L}_\nu^{(\alpha)}(x)$ is the generalized Laguerre polynomial.
\end{lemma}

\begin{proof}
 $A_{ij}$ is a random variable defined by: \[A_{ij}=\lVert X_{i}+\theta_{i}+\varepsilon_{i}-X_{j}-\theta_{j}-\varepsilon_j \rVert_2=\sqrt{\sum_{k=1}^p \left( x_{ik}+\theta_{ik}+\varepsilon_{ik}-x_{jk}-\theta_{jk}-\varepsilon_{jk}\right)^2}.\]

\noindent The random variable $\displaystyle x_{ik}+\theta_{ik}+\varepsilon_{ik}-x_{jk}-\theta_{jk}-\varepsilon_{jk}$ is normally distributed as $ \mathcal{N}(x_{ik}+\theta_{ik}-x_{jk}-\theta_{jk},2\sigma^2)$
which implies that the random variable 
\[\displaystyle \sqrt{\sum_{k=1}^p \left(\frac{ x_{ik}+\theta_{ik}+\varepsilon_{ik}-x_{jk}-\theta_{jk}-\varepsilon_{jk}}{\sqrt{2}\sigma}\right)^2}\]  is distributed according to the non-central chi distribution with $p$ degrees of freedom and $\lambda_{ij}$ the non-centrality parameter that is related to the mean of the random variable by:  $\displaystyle \lambda_{ij}=\sqrt{\sum_{k=1}^p \left(\frac{ x_{ik}+\theta_{ik}-x_{jk}-\theta_{jk}}{\sqrt{2}\sigma}\right)^2}$.   

\noindent Then, we obtain:
\[\displaystyle \mathbb{E}\left( \sqrt{\sum_{k=1}^p \left(\frac{ x_{ik}+\theta_{ik}+\varepsilon_{ki}-x_{jk}-\theta_{jk}-\varepsilon_{jk}}{\sqrt{2}\sigma}\right)^2}\right)= \frac{1}{\sqrt{2}\sigma}\mathbb{E}(A_{ij}).\]

\noindent Recall that the expectation of non-central chi distribution $\chi_p(\lambda_{ij})$ is given by:
\[\sqrt{\frac{\pi}{2}} \mathcal{L}_{\frac{1}{2}}^{\frac{p}{2}-1}\left(-\frac{\lambda_{ij}^2}{2}\right).\]

\noindent We note \[\mu_{ij}=\sqrt{\frac{\pi}{2}} \mathcal{L}_{\frac{1}{2}}^{\frac{p}{2}-1}\left(-\frac{\lambda_{ij}^2}{2}\right).\] 
Then, we have \[\mathbb{E}(A_{ij})=\sqrt{2}\sigma \mu_{ij}.\]

\noindent Moreover, we have:
\[\displaystyle \mathbb{V}ar\left( \sqrt{\sum_{k=1}^p \left(\frac{ x_{ik}+\theta_{ik}+\varepsilon_{ik}-x_{jk}-\theta_{jk}-\varepsilon_{jk}}{\sqrt{2}\sigma}\right)^2}\right)=\frac{1}{{2}\sigma^2}\mathbb{V}ar(A_{ij})\cdot\]
Recall that the variance of non-central chi distribution $\chi_p(\lambda_{ij})$ is given by:
\[p+\lambda_{ij}^{2}-\mu_{ij}^{2}\]
that gives: \[\mathbb{V}ar(A_{ij})={2}\sigma^2(p+\lambda_{ij}^{2}-\mu_{ij}^{2}).\]
\\

\noindent Concerning the calculation of moments of order  $3$ and $4$, we have: 
\[\displaystyle \mathbb{E}\left( \sqrt{\sum_{k=1}^p \left(\frac{ x_{ik}+\theta_{ik}+\varepsilon_{ik}-x_{jk}-\theta_{jk}-\varepsilon_{jk}}{\sqrt{2}\sigma}\right)^2}\right)^3= \frac{1}{2\sqrt{2}\sigma^3}\mathbb{E}(A_{ij}^3)\]

\noindent and
\[\displaystyle \mathbb{E}\left( \sqrt{\sum_{k=1}^p \left(\frac{ x_{ik}+\theta_{ik}+\varepsilon_{ik}-x_{jk}-\theta_{jk}-\varepsilon_{jk}}{\sqrt{2}\sigma}\right)^2}\right)^4= \frac{1}{4\sigma^4}\mathbb{E}(A_{ij}^4)\cdot\]
For a non-central chi distribution, the moments $3$ and $4$ are, respectively, given by:
\[3\sqrt{\frac{\pi}{2}} \mathcal{L}_{\frac{3}{2}}^{\frac{p}{2}-1}\left(-\frac{\lambda_{ij}^{2}}{2}\right) \text{ and } (p+\lambda_{ij}^{2})^2+2(p+2\lambda_{ij}^{2}).\]
 
\noindent Then, we obtain: 
\[\mathbb{E}(A_{ij}^3)=6\sigma^3\sqrt{{\pi}} \mathcal{L}_{\frac{3}{2}}^{\frac{p}{2}-1}(-\frac{\lambda_{ij}^{2}}{2}),\] and \[\mathbb{E}(A_{ij}^4)=4\sigma^4(p+\lambda_{ij}^{2})^2+8\sigma^4(p+2\lambda_{ij}^{2})\cdot\]

\noindent  Moreover, we can straightforwardly find:
\[\mathbb{E}(A^2{_{ij}})=2\sigma^2(p+\lambda_{ij}^{2}) \text{ \hspace{1cm} and \hspace{1cm} } \mathbb{V}ar(A^2{_{ij}})=8\sigma^4(p+2\lambda_{ij}^{2})   .\]
\end{proof}

\begin{lemma} \label{lemme3}
Upper bounds of $\mathbb{E}(A_{ij}A_{ij'})$ and $\mathbb{E}(A_{ij}^2 A_{ij'}^2)$ are given by:
\[\mathbb{E}(A_{ij}A_{ij'})\leq 2\sigma^2\sqrt{(p+\lambda^2_{ij})(p+\lambda^2_{ij'})}, \]
\[\mathbb{E}(A_{ij}^2 A_{ij'}^2)\leq 4\sigma^4\sqrt{\left[ (p+\lambda^2_{ij})^2+2(p+2\lambda^2_{ij})\right]\left[ (p+\lambda^2_{ij'})^2+2(p+2\lambda^2_{ij'})\right]}\cdot \]
\end{lemma}

\begin{proof}

\noindent The variables $A_{ij}$ and $A_{ij'}$ are two dependent random variables.
\noindent Using the Cauchy-Schwartz inequality, we can write:
\[ \mathbb{E}(A_{ij}A_{ij'})\leq \sqrt{\mathbb{E}(A_{ij}^2)\mathbb{E}(A_{ij'}^2)} \cdot\]  

\noindent Using Lemma $\ref{lemma2}$, we obtain:
\[\mathbb{E}(A_{ij}A_{ij'})\leq 2\sigma^2\sqrt{(p+\lambda^2_{ij})(p+\lambda^2_{ij'})}\cdot \]

\noindent Moreover,
\begin{eqnarray}
\mathbb{E}(A_{ij}^2A_{ij'}^2)\leq \sqrt{\mathbb{E}(A_{ij}^4)\mathbb{E}(A_{ij'}^4)}\cdot\nonumber
\end{eqnarray}

\noindent And from Lemma $\ref{lemma2}$, we obtain:
\[\mathbb{E}(A_{ij}^2 A_{ij'}^2)\leq 4\sigma^4\sqrt{\left[ (p+\lambda^2_{ij})^2+2(p+2\lambda^2_{ij})\right]\left[ (p+\lambda^2_{ij'})^2+2(p+2\lambda^2_{ij'})\right]} \cdot \]
\end{proof}

\begin{lemma} \label{lemma5}
Lower bounds of \ $\mathbb{V}ar(A_{ij}+A_{ij'})$, $cov(A_{ij},A_{ij'})$ and $\mathbb{E}(A_{ij}A_{ij'})$ are given by:
\[\mathbb{V}ar(A_{ij}+A_{ij'})\geq 2\sigma^2\left( p+\lambda_{jj'}^2-(\mu_{ij}+\mu_{ij'})^2 \right),\]

\[cov(A_{ij},A_{ij'})\geq  -\sigma^2 \left( p+2 \mu_{ij} \mu_{ij'}+ \lambda_{ij}^2+\lambda_{ij'}^2-\lambda_{jj'}^2 \right),\]

\[\mathbb{E}(A_{ij}A_{ij'}) \geq -\sigma^2 \left( p+ \lambda_{ij}^2+\lambda_{ij'}^2-\lambda_{jj'}^2 \right) \cdot\]
\end{lemma}

\begin{proof}
We have:
\begin{eqnarray}
A_{ij}+A_{ij'}&=&\lVert X_i+\theta_i+\varepsilon_i- X_j-\theta_j-\varepsilon_j \rVert+\lVert X_{i}+\theta_{i}+\varepsilon_{i}- X_{j'}-\theta_{j'}-\varepsilon_{j'} \rVert\nonumber\\
&\geq & \lVert X_j+\theta_j+\varepsilon_j- X_{j'}-\theta_{j'}-\varepsilon_{j'} \rVert=A_{jj'}\cdot\label{eqE12}
\end{eqnarray}

\noindent The variance of $A_{ij}+A_{ij'}$ is given by:
\[\mathbb{V}ar(A_{ij}+A_{ij'})=\mathbb{E}(A_{ij}+A_{ij'})^2-\left(\mathbb{E}(A_{ij}+A_{ij'})\right)^2\cdot\]
\noindent Using inequality $\ref{eqE12}$ and the positively of $A_{ij}$,  we have $(A_{ij}+A_{ij'})^2 \geq A_{jj'}^2$ and then:
\[ \mathbb{E}\left( (A_{ij}+A_{ij'})^2 \right) \geq \mathbb{E} (A_{jj'}^2)=2\sigma^2(p+\lambda_{jj'}^2)\cdot\]
\noindent Hence, we have:
\begin{eqnarray}
\mathbb{V}ar(A_{ij}+A_{ij'})&\geq& 2\sigma^2(p+\lambda_{jj'}^2)-\left(\mathbb{E}(A_{ij})+ \mathbb{E}(A_{ij'}) \right)^2\nonumber\\
&\geq& 2\sigma^2(p+\lambda_{jj'}^2)- 2\sigma^2(\mu_{ij}+\mu_{ij'})^2\nonumber\\
&\geq& 2\sigma^2\left( p+\lambda_{jj'}^2- (\mu_{ij}+\mu_{ij'})^2\right)\cdot\nonumber
\end{eqnarray}
\vspace{3mm}

\noindent To obtain the lower bound of $cov(A_{ij},A_{ij'})$, we use the definition: 
\[cov(A_{ij},A_{ij'})=\frac{1}{2} \left[ \mathbb{V}ar(A_{ij}+A_{ij'})-\mathbb{V}ar(A_{ij})-\mathbb{V}ar(A_{ij'}) \right]\cdot\]
Then, using Lemma \ref{lemma2} and the above result, we obtain: 
\begin{eqnarray}
cov(A_{ij},A_{ij'})&\geq& \sigma^2(p+\lambda_{jj'}^2- (\mu_{ij}+\mu_{ij'})^2)-\sigma^2(p+\lambda_{ij}^2-\mu_{ij}^2)-\sigma^2(p+\lambda_{ij'}^2-\mu_{ij'}^2) \nonumber\\
&\geq& -\sigma^2 \left( p+2 \mu_{ij} \mu_{ij'}+ \lambda_{ij}^2+\lambda_{ij'}^2-\lambda_{jj'}^2 \right)\cdot \label{eqc32}
\end{eqnarray}

\noindent Concerning the lower bound of $\mathbb{E}(A_{ij}A_{ij'})$, we have:
\begin{eqnarray}
\mathbb{E}(A_{ij}A_{ij'})&=&cov(A_{ij},A_{ij'})+\mathbb{E}(A_{ij}) \mathbb{E}(A_{ij'})\nonumber\\
&\geq& -\sigma^2 \left( p+2 \mu_{ij} \mu_{ij'}+ \lambda_{ij}^2+\lambda_{ij'}^2-\lambda_{jj'}^2 \right)+2\sigma^2 \mu_{ij} \mu_{ij'}\nonumber\\
&\geq& -\sigma^2 \left( p+ \lambda_{ij}^2+\lambda_{ij'}^2-\lambda_{jj'}^2 \right)\cdot\nonumber
\end{eqnarray}
\end{proof}

\begin{lemma} \label{lemma6}
Lower bounds of $\mathbb{E}(A_{ij}^2 A_{ij'})$ and $\mathbb{E}( A^2_{ij'}A_{ij})$ are given by:
\[\mathbb{E}(A^2_{ij}A_{ij'}) \geq -\sigma^2 \left( p+ \lambda_{ij}^2+\lambda_{ij'}^2-\lambda_{jj'}^2 \right) -\sqrt{2}\sigma \mu_{ij'},\]
\[\mathbb{E}(A^2_{ij'}A_{ij}) \geq -\sigma^2 \left( p+ \lambda_{ij}^2+\lambda_{ij'}^2-\lambda_{jj'}^2 \right) -\sqrt{2}\sigma \mu_{ij}\cdot\]
\end{lemma}

\begin{proof}
\begin{eqnarray}
A_{ij}^2A_{ij'} &\geq &(A_{ij}-1)A_{ij'} \text{ as } A^2_{ij} \geq A_{ij}-1\nonumber\\
&\geq& A_{ij}A_{ij'}-A_{ij'}\cdot\nonumber
\end{eqnarray}

\noindent Thus, \[\mathbb{E}\left( A_{ij}^2A_{ij'}\right) \geq \mathbb{E} (A_{ij}A_{ij'})-\mathbb{E}( A_{ij'})\cdot\]
Lemma $\ref{lemma5}$ leads:
\[\mathbb{E}(A^2_{ij}A_{ij'}) \geq -\sigma^2 \left( p+ \lambda_{ij}^2+\lambda_{ij'}^2-\lambda_{jj'}^2 \right) -\sqrt{2}\sigma \mu_{ij'}\cdot\]

\noindent Similarly, we obtain:
\[\mathbb{E}(A^2_{ij'}A_{ij}) \geq -\sigma^2 \left( p+ \lambda_{ij'}^2+\lambda_{ij}^2-\lambda_{jj'}^2 \right) -\sqrt{2}\sigma \mu_{ij}\cdot\]
\end{proof}

\section{Calculation of the expectation value of error $\Delta$} \label{App2}

\noindent Using equations (\ref{eqdef1}) and (\ref{eqt2}), we have:
 \begin{eqnarray}
e_{ij}&=&d^2_{{ij}}+a^2A_{ij}^2-2 a  d_{ij} A_{ij}\cdot \label{eqt3}
\end{eqnarray}

\noindent We suppose that the $p$ components of vectors $\varepsilon_i$, for all $i=1,\dots,n $ are  identically independent random variables and normally distributed. Using Lemma \ref{lemma2}, we obtain: \[\mathbb{E}(e_{ij})=d_{ij}^2 +2a^2\sigma^2(p+\lambda_{ij}^{2})-2\sqrt{\pi}a \sigma  d_{ij} \mathcal{L}_{\frac{1}{2}}^{\frac{p}{2}-1}\left(-\frac{\lambda_{ij}^2}{2}\right)\cdot\] 
where $\displaystyle \lambda_{ij}=\frac{1}{\sqrt{2}\sigma}\lVert X_i+\theta_i-X_j-\theta_j\rVert_{_2} \cdot$

\noindent Hence, the expectation of error $\Delta$ is equal to: 

\begin{equation} \displaystyle \mathbb{E}(\Delta)=\sum_{1\leq i<j\leq n} \left[ d_{ij}^2 +2a^2\sigma^2(p+\lambda^2_{ij})-2\sqrt{\pi} a \sigma  d_{ij} \mathcal{L}_{\frac{1}{2}}^{\frac{p}{2}-1}\left(-\frac{\lambda_{ij}^2}{2}\right) \right]\cdot \label{eqEsp}\end{equation}

\section{Calculation of variance value of error $\Delta$} \label{App3}
\noindent The variance of $\Delta$ is given by: 
\begin{eqnarray}
\mathbb{V}ar(\Delta)&=&\mathbb{V}ar(\displaystyle \sum_{1\leq i<j\leq n} e_{ij})\nonumber\\
&=&\sum_{1\leq i<j\leq n} \mathbb{V}ar(e_{ij}) +{ \sum_{\begin{array}{c}
1\leq i<j\leq n\\
1\leq i'<j'\leq n\\
\end{array}}}  cov(e_{ij}; e_{i'j'})\nonumber
\end{eqnarray}

\noindent As  $cov(e_{ij},e_{i'j'})=0$,  if $(i,j)\cap(i',j')=\emptyset$ we obtain:
\begin{equation}\label{eqforVar}
\mathbb{V}ar(\Delta)=\sum_{1\leq i<j\leq n} \mathbb{V}ar(e_{ij}) +2\sum_{1\leq i <j<j'\leq n} cov(e_{ij},e_{ij'}) \cdot 
\end{equation}



\noindent To calculate $\mathbb{V}ar(\Delta)$ it is necessary to calculate $\mathbb{V}ar(e_{ij})$ and $cov(e_{ij}; e_{ij'})$ for all couples $(i,j)$ and $(i,j')$ with $
1\leq i<j<j'\leq n$. 
\subsection{Calculation of $\mathbb{V}ar(e_{ij})$}
We have from Equation (\ref{eqt3}):
 \begin{eqnarray}
e_{ij}&=&d^2_{{ij}}+a^2A_{ij}^2-2 a d_{ij} A_{ij}\cdot \nonumber
\end{eqnarray}
\noindent  The definition of variance is:
\begin{eqnarray}
\mathbb{V}ar(e_{ij})=\mathbb{E}(e_{{ij}}^2)-(\mathbb{E}(e_{ij}))^2 \cdot\label{eqt5}
\end{eqnarray}
Let begin by the calculation of  $\mathbb{E}(e_{{ij}}^2)$.
\begin{eqnarray}
e_{{ij}}^2&=&\left(d_{{ij}}^2+a^2A_{{ij}}^2-2ad_{ij}A_{ij}\right)^2\nonumber\\
&=& d_{{ij}}^4+a^4A_{ij}^4+6a^2d_{ij}^2A_{ij}^2-4ad_{ij}^3A_{ij}-4a^3d_{ij}A_{ij}^3\cdot \label{eqES}\nonumber
\end{eqnarray}

\noindent The expectation of $e_{ij}^2$ is then given by:
\begin{equation}
\mathbb{E}(e_{{ij}}^2)=d_{{ij}}^4+a^4\mathbb{E}( A_{ij}^4)+6a^2d_{ij}^2\mathbb{E}(A_{ij}^2)-4ad_{ij}^3\mathbb{E}(A_{ij})-4a^3d_{ij}\mathbb{E}(A_{ij}^3)\cdot \label{eqvar}
\end{equation}
\noindent Using lemma \ref{lemma2}, we can obtain all the terms of the moments presented in Eequation (\ref{eqvar}).
Then, we obtain the variance by replacing each term with their value in Equation $(\ref{eqt5})$ and we obtain:
\begin{eqnarray}
\mathbb{V}ar(e_{ij})=&a^4\mathbb{E}( A_{ij}^4)+4a^2d_{ij}^2\mathbb{E}(A_{ij}^2)-4a^3d_{ij}\mathbb{E}(A_{ij}^3)-a^4(\mathbb{E}(A_{ij}^2))^2-4 a^2 d_{ij}^2 (\mathbb{E}(A_{ij}))^2\nonumber\\
&+4a^3 d_{ij}\mathbb{E}(A_{ij}^2) \mathbb{E}(A_{ij})\cdot\nonumber
\end{eqnarray}

\subsection{Calculation of $cov(e_{ij}, e_{ij'})$}
\noindent Now, we want to calculate $cov(e_{ij}, e_{ij'})$. The definition of the covariance is given by:
\begin{eqnarray}
cov(e_{ij}, e_{ij'})=\mathbb{E}(e_{ij} e_{ij'})-\mathbb{E}(e_{ij}) \mathbb{E}(e_{ij'})\cdot \label{eq1}
\end{eqnarray}
To calculate the expectation $\mathbb{E}(e_{ij} e_{ij'})$, we firstly calculate $e_{ij}e_{ij'}$: 
\begin{eqnarray}
e_{ij}e_{ij'}&=&\left( d_{ij}^2+a^2A_{ij}^2-2ad_{ij}A_{ij}\right)\left( d_{ij'}^2+a^2A_{ij'}^2-2ad_{ij'}A_{ij'}\right)\nonumber\\
&=&d_{ij}^2d_{ij'}^2+a^2d_{ij}^2A_{ij'}^2-2ad_{ij}^2d_{ij'}A_{ij'}
+a^2d_{ij'}^2A_{ij}^2+a^4A_{ij}^2A_{ij'}^2\nonumber\\
&&-2a^3d_{ij'}A_{ij}^2A_{ij'}-2ad_{ij}d_{ij'}^2A_{ij}-2a^3d_{ij}A_{ij}A_{ij'}^2+4a^2d_{ij}d_{ij'}A_{ij}A_{ij'} \cdot\nonumber
\end{eqnarray}

\noindent Passing to the expectation, we obtain:
\begingroup\makeatletter\def\f@size{11}\check@mathfonts
\begin{eqnarray}
\mathbb{E} (e_{ij}e_{ij'})=d_{ij}^2d_{ij'}^2+a^2d_{ij}^2\mathbb{E}( A_{ij'}^2)-2ad_{ij}^2d_{ij'}\mathbb{E}(A_{ij'})
+a^2d_{ij'}^2\mathbb{E}(A_{ij}^2)+a^4\mathbb{E}(A_{ij}^2A_{ij'}^2)\nonumber\\
 -2a^3d_{ij'}\mathbb{E}(A_{ij}^2A_{ij'})-2ad_{ij}d_{ij'}^2\mathbb{E}(A_{ij})-2a^3d_{ij}\mathbb{E}(A_{ij}A_{ij'}^2)+4a^2d_{ij}d_{ij'}\mathbb{E}(A_{ij}A_{ij'})\cdot \nonumber\\\label{eqCov}
\end{eqnarray}
\endgroup

\noindent Using the five lemmas presented in section $\ref{App1}$, we can  bound $\mathbb{E}(e_{ij}e_{ij'})$  in a way to obtain an upper bound of the covariance $cov(e_{ij}e_{ij'})$. We note $B_{ijj'}$ the  upper bound of $\mathbb{E}(e_{ij}e_{ij'})$. So, returning to Equation ($\ref{eqforVar}$), we obtain:
\[\mathbb{V}ar(\Delta)\leq \sum_{1\leq i<j\leq n} \mathbb{E}(e_{ij}^2)-\mathbb{E}(e_{ij})^2 +2\sum_{1\leq i <j<j'\leq n} B_{ijj'}-2\sum_{1\leq i <j<j'\leq n} \mathbb{E}(e_{ij})\mathbb{E}(e_{ij'})\cdot\]

\end{appendices}
\end{document}